\documentclass[onecolumn,12pt,english]{IEEEtran}
\usepackage[T1]{fontenc}
\usepackage[latin9]{inputenc}
\usepackage{amsmath}
\usepackage{amsfonts}
\usepackage{graphicx}
\usepackage{amssymb}
\usepackage{epstopdf}
\usepackage{enumerate}
\usepackage{amsthm}
\usepackage{bm}
\usepackage{outlines}
\usepackage{midfloat}
\usepackage{babel}
\usepackage{algpseudocode}
\usepackage{algorithm}
\usepackage{enumerate}
\usepackage{caption}
\usepackage{makeidx}
\usepackage{subfigure}
\usepackage{float}
\makeatletter
\usepackage{bbding}\usepackage{pifont}
\makeatother

\newtheorem{lemma}{Lemma}
\newtheorem{theorem}{Theorem}
\makeindex

\begin{document}

\title{Millimeter-Wave Beamformed Full-dimensional MIMO Channel Estimation Based on Atomic Norm Minimization}

%\author{Yingming Tsai \ \ \ Le Zheng \ \ \ Xiaodong Wang   \\
%Electrical Engineering Dept., Columbia Univ., New York, NY 10027}

%\author{\IEEEauthorblockN{Author One\IEEEauthorrefmark{1},
%		Author Two\IEEEauthorrefmark{2},  and
%		Author Four\IEEEauthorrefmark{3}}
%\IEEEauthorblockA{Department of Whatever,
%	Whichever University\\
%	Email: \IEEEauthorrefmark{1}author.one@add.on.net,
%	\IEEEauthorrefmark{2}author.two@add.on.net,
%	\IEEEauthorrefmark{3}author.three@add.on.net}	
%}

\author{Yingming~Tsai,~\IEEEmembership{Member,~IEEE,}
	Le~Zheng,~\IEEEmembership{Member,~IEEE,}
	and~Xiaodong~Wang,~\IEEEmembership{Fellow,~IEEE}\\
	Electrical Engineering Dept., Columbia Univ., New York, NY 10027}

\maketitle

\begin{abstract}
The millimeter-wave (mmWave)  full-dimensional (FD) MIMO system employs planar arrays at both the base station and user equipment and can simultaneously support both azimuth and elevation beamforming. In this paper, we propose atomic-norm-based methods for mmWave FD-MIMO channel estimation under both uniform planar arrays (UPA) and non-uniform planar arrays (NUPA). Unlike existing algorithms such as compressive sensing (CS) or subspace methods, the atomic-norm-based algorithms do not require to discretize the angle spaces of the angle of arrival (AoA) and angle of departure (AoD) into grids, thus provide much better accuracy in estimation. In the UPA case, to reduce the computational complexity, the original large-scale 4D atomic norm minimization problem is approximately reformulated as a semi-definite program (SDP) containing two decoupled two-level Toeplitz matrices. The SDP is then solved via the alternating direction method of multipliers (ADMM) where each iteration involves only closed-form computations. In the NUPA case, the atomic-norm-based formulation for channel estimation becomes nonconvex and a gradient-decent-based algorithm is proposed to solve the problem. Simulation results show that the proposed algorithms achieve better performance than the CS-based and subspace-based algorithms.

\end{abstract}
\textbf{Keywords:}   Full-dimensional (FD) MIMO, uniform planar  array (UPA), non-uniform planar  array (NUPA), atomic norm, channel estimation, millimeter-wave, alternating direction method of multipliers (ADMM), gradient descent. 

\section{Introduction\label{Section_Intro}}
Millimeter wave (mmWave) communications have been proposed as an important physical-layer technology for the $5$th generation (5G) mobile networks to provide multi-gigabit  services \cite{5G_it_work}.  Two prominent features of the mmWave spectrum are
the massive bandwidth available and the tiny wavelengths compared to conventional microwave bands, thus enabling dozens or even
hundreds of antenna elements to be accommodated at communication link ends
within a reasonable physical form factor.  This suggests that massive MIMO and mmWave technologies should be considered jointly to provide higher data rates and spectrum efficiency. In particular,  the mmWave full-dimensional  MIMO (FD-MIMO)  systems \cite{3D-MIMO},\cite{mmWave_BF_for_backhaul_and_small_cell}
employ uniform or non-uniform
planar arrays
 at both the basestation (BS) and user equipment (UE) and provide an extra degree of freedom in
the elevation-angle domain. Users can now be distinguished not only by their
AoAs in the azimuth  domain  but also by their AoDs in the elevation domain \cite{FD-MIMO_next_gen_cell_tech}.  In this paper, we consider channel estimation for mmWave FD-MIMO systems that simultaneously support both azimuth and elevation beamforming.   

The mmWave band  channel is significantly different from those in sub-6GHz bands. The key challenge in designing new radio access technologies for mmWave is how to overcome the much larger path-loss and reduce blockage probability. To that end, beamforming is essential in combating the serve path-loss for wireless system operating in mmWave bands \cite{Beamforming_Survey}. However, to estimate the full channel state information (CSI) under beamformed FD-MIMO is somehow challenging because the receiver typically only  obtains the beamformed CSI instead of full CSI. To address this issue, fast beam
scanning and searching techniques have been extensively studied \cite{mmWave_BF_for_backhaul_and_small_cell,Beam_codebook_vased_beamforming_protocol_for_mmWave}. The objective of beam scanning is to search for the best beamformer-combiner pair by letting the transmitter and receiver scan the adaptive sounding beams
chosen from pre-determined sounding beam codebooks.  However, the exhaustive search may be hampered by the high training overhead in practice and suffer from low spectral efficiency. Another approach is to  estimate the mmWave channel or its associated parameters, by exploiting the sparse scattering nature of the mmWave channels \cite{3D_mmWave_channel_model},\cite{3D_mmWave_channel_model_proposal},\cite{Spatially_Sparse-Precoding_in_MM_Wave}, that is,  mmWave channel estimation can be formulated as a sparse signal recovery problem \cite{Channel_Estimation_and_Hybrid_Precoding}, \cite{CS_based_multi_users_mmWave}, \cite{2D_MUSIC_mmWave},  \cite{Millimeter_Wave_MIMO_channel_estimation_Adaptive_CS} and solved using the compressive sensing (CS)-based approach \cite{Donoho_ComprssedSensing}.  In the CS-based approach,  a  sensing matrix needs to be constructed first, by dividing certain parameter space into a finite number of grids and  thus the channel estimation performance is limited by the grid resolution.   On the other hand, in \cite{book_mmWav_Communications},   a subspace-based mmWave MIMO channel estimation method  that makes use of the MUSIC algorithm is proposed.  A  2D-MUSIC algorithm for beamformed mmWave MIMO channel estimation is proposed in \cite{2D_MUSIC_mmWave} to further enhance the channel estimation performance. The MUSIC algorithm is able to identify multiple paths with high resolution but it is sensitive to antenna position, gain, and phase errors.

Recently, the atomic norm minimization  \cite{Compressed_Sensing_Off_the_Grid} has been applied to many signal processing problems such as  super-resolution 
frequency estimation \cite{Bhaskar_AtomicNormDenoise}, \cite{Gridless_super_resolution_direction_finding}, spectral estimation \cite{Atomic_Norm_Denoising}, AoA estimation,  \cite{ANM_DoA_estimation,ANM_2D_DoA_estimation}, uplink multiuser MIMO channel estimation \cite{ANM_MIMO_CHEST} and linear system identification
\cite{ANM_Linear_System_Identification}.   Under certain conditions, atomic norm minimization can achieve exact frequency localization, avoiding the effects of basis mismatch which can plague grid-based CS techniques. Different from the prior works such as CS-based and subspace-based channel estimation methods mentioned above, we  formulate the mmWave FD-MIMO channel estimation as an atomic norm minimization problem. Unlike \cite{ANM_MIMO_CHEST} that considers uplink multiuser MIMO channel estimation, in which the uniform linear array is assumed and only the AoA parameter is estimated, in this paper, we consider the mmWave beamformed FD-MIMO channel, which involves the estimation of both AoA and AoD. Hence, instead of one-dimensional (1D) atomic norm minimization, our problem is formulated as a four-dimensional (4D) atomic norm minimization problem. The 4D\ atomic norm minimization can be transformed into semi-definite program (SDP) which is of high dimensional and involves block Toeplitz  matrices, leading to very high computational complexity.  Therefore, we introduce a 4D atomic norm  approximation method to reduce the computational complexity and an efficient algorithm  based on the alternating direction method of multipliers (ADMM)  is derived. 

Recently, non-uniform planar array (NUPA) has attracted more interest due to its ability in reducing sidelobes and antenna correlation \cite{structured_non_unifromly_spaced_antenna_array,NUPA_mmwave}. NUPA can potentially increase the achievable multiplexing gain of mmWave FD-MIMO beamforming. However, the corresponding atomic norm minization problem cannot be transformed into an SDP  when the antennas are not uniformly placed \cite{Compressed_Sensing_Off_the_Grid}. Hence, we propose a gradient descent method for mmWave FD-MIMO channel estimation with NUPA.

The remainder of the paper is organized as follows. In Section \ref{Section_SystemModel}, the mmWave beamformed FD-MIMO channel model   is introduced. In Section \ref{Section_chest}, we formulate the mmWave FD-MIMO channel estimation as an atomic norm minimization problem for the case of UPA. In Section \ref{SEC_AN_APPROX}, we develop efficient algorithms for implementing the proposed atomic-norm-based channel estimator.  In Section \ref{Section_NUPA_chest}, we consider the case of NUPA\ and provide the formulation and algorithm for the atomic-norm-based channel estimator. In Section \ref{Section_Simulation}, simulation results are provided.  Finally, Section \ref{Section_Conclusions} concludes the paper. 

\section{System Descriptions and Background \label{Section_SystemModel}}
\subsection{System and Channel Models}
We consider a mmWave FD-MIMO system with $M$ receive antennas and $N$ transmit antennas that simultaneously supports elevation and azimuth beamforming. The channel matrix can be expressed in terms of transmit and receive array responses \cite{Spatially_Sparse-Precoding_in_MM_Wave}:
\begin{eqnarray} \label{Eq_2_2}
\mathbf H = \mathbf B \mathbf \Sigma \mathbf A^{H} = \sum_{l=1}^{L}\sigma_{l}
\mathbf b(\mathbf f_{l})  \mathbf a( \mathbf g_{l} )^{H}, 
\end{eqnarray} 
where $\left( \cdot \right)^{H}$ denotes the Hermitian transpose; the matrix $\mathbf \Sigma =\text{diag}(\bm{\sigma})=\text{diag}\left(\left[
\sigma_{1} \ \sigma_{2} \ldots \sigma_{L} \right]^{T} \right)$ is a diagonal
matrix with each $\sigma_{l} \in \mathbb{C}$ denoting the $l$-th multipath gain; $L$ denotes the number of paths; the matrices $\mathbf B = \left[ \mathbf b(\mathbf f_{1}) \ldots \mathbf b(\mathbf f_{L}) \right]$ and $\mathbf A = \left[ \mathbf a(\mathbf g_{1}) \ldots \mathbf a(\mathbf g_{L}) \right]$ denote the steering responses of the receive and transmit arrays, respectively. For a linear array with half-wavelength separation of adjacent antenna elements, the array response is in the form of a uniformly sampled complex sinusoid with frequency $x \in [-\frac{1}{2},\ \frac{1}{2})$: 
\begin{eqnarray} \label{Eq_uniform_sin}
\mathbf c_{n} \left( x \right) = \frac{1}{\sqrt{n}} \left[ 1\ e^{j2\pi x} \cdots e^{j2\pi \left( n-1 \right) x } \right]^{T} \in \mathbb{C}^{n \times 1}.
\end{eqnarray}
We assume that both the transmitter (Tx) and receiver (Rx)\ are equipped with uniformly spaced planar antenna arrays (UPA)s \cite{FD-MIMO_Using_Uniform_Planar_Arrays,3D_MIMO_ULA}, each with half-wavelength antenna element separations along the elevation-and-azimuth-axis. Then the Tx and Rx array responses can be expressed as \cite{3D_MIMO_ULA}
\begin{eqnarray} \label{Eq_2_3}
\mathbf a(\mathbf g_{l}) &=& \mathbf c_{N_{1}}\left( g_{l,1} \right) \otimes \mathbf c_{N_{2}}\left( g_{l,2} \right), \\ \label{Eq_2_3_b}
\mathbf b\left( \mathbf f_{l} \right) &=& \mathbf c_{M_{1}}\left( f_{l,1} \right) \otimes \mathbf c_{M_{2}}\left( f_{l,2} \right),
\end{eqnarray}
with
\begin{eqnarray} \label{Eq_g_f_definition}
\mathbf g_{l}&=&\left\{ g_{l,1} =\frac{1}{2} \sin\left(
\vartheta_{l}\right)\cos\left( \varphi_{l} \right), g_{l,2}= \frac{1}{2}\cos\left( \vartheta_{l} \right) \right\}, \\ 
\mathbf f_{l}&=&\left\{ f_{l,1} =\frac{1}{2} \sin\left(
\theta_{l}\right)\cos\left( \phi_{l} \right), f_{l,2}= \frac{1}{2}\cos\left(
\theta_{l} \right) \right\},
\end{eqnarray}
where  $\otimes$ denotes the Kronecker product; $\vartheta_{l}$, $\varphi_{l}$ denote elevation and azimuth angles of the angle of departure (AoD) of the $l$-th path, respectively; and $\theta_{l}$, $\phi_{l}$ denote elevation and azimuth angles of the angle of arrival (AoA), respectively. Here, $N_{1}$, $N_{2}$ denote the numbers of elevation and azimuth transmit antennas, respectively, and the total number of transmit antennas is $N=N_{1}N_{2}$. Similarly, $M_{1}$, $M_{2}$ denote the numbers of elevation and azimuth receive antennas, respectively, and the total number of receive antennas is $M=M_{1} M_{2}$. For the UPA configuration, it can resolve the AoA and AoD in $360^{\circ}$ range, thereby $\vartheta_{l}, \theta_{l}, \varphi_{l}, \phi_{l} \in [-\pi,\ \pi]$ and $g_{l,1} = \frac{1}{2}\sin\left(
\vartheta_{l}\right)\cos\left( \varphi_{l} \right) \in [-\frac{1}{2},\ \frac{1}{2}) $, $g_{l,2} = \frac{1}{2}\cos\left( \vartheta_{l} \right) \in [-\frac{1}{2},\ \frac{1}{2})$, $f_{l,1} = \frac{1}{2}\sin\left(
\theta_{l}\right)\cos\left( \phi_{l} \right) \in [-\frac{1}{2}, \frac{1}{2})$, $f_{l,2}  = \frac{1}{2}\cos\left( \theta_{l} \right) \in [-\frac{1}{2}, \frac{1}{2})$.

To estimate the channel matrix, the transmitter transmits $P$ distinct beams during $P$ successive time slots. i.e., in the $p$-th time slot, the beamforming vector $\mathbf p_{p} \in \mathbb{C}^{N\times 1}$ is selected from a set of unitary vectors in the form of Kronecker-product-based codebook, e.g., $\mathbf p_{p} = \mathbf p_{p,1} \otimes \mathbf p_{p,2}$ where $\mathbf p_{p,1} \in \mathbb{C}^{N_{1}}$ and  $\mathbf p_{p,2}\in \mathbb{C}^{N_{2}}$ are selected from two DFT codebooks of dimensions $N_1$ and $N_2$, respectively \cite{DFT_beamforming}. The $p$-th received signal vector can be expressed as
\begin{equation} \label{Eq_2_1}
\mathbf{y}_{p} = \mathbf H\mathbf p_{p} s_{p} +\mathbf w_{p}, 
\end{equation}
where  $\mathbf w_{p} \sim\ \mathcal{CN}\left(\mathbf 0, \sigma^{2}_{w}\mathbf{I}_{M} \right)$ is the additive white Gaussian noise (AWGN) with $\mathbf I_{M}$ denoting the $M\times M$ identity matrix, and $s_{p}$ denotes the pilot symbol in the $p$-th time slot. The receiver collects $\mathbf y_{p} \in \mathbb{C}^{M \times 1}$ for $p=1, \ldots P$ and concatenates them to obtain the signal matrix 
\begin{eqnarray} \label{Eq_2_4}
\mathbf Y = \left[ \mathbf y_{1}\ \mathbf y_{2} \ldots \mathbf y_{P} \right] = \mathbf H \mathbf P \mathbf S + \mathbf W = \mathbf B \mathbf \Sigma \mathbf A^{H} \mathbf P \mathbf S +\mathbf W,
\end{eqnarray} 
where $\mathbf P = \left[ \mathbf p_{1}\ \mathbf p_{2} \ldots \mathbf p_{P} \right] \in \mathbb{C}^{N \times P}$,  $\mathbf W = \left[ \mathbf w_{1}\ \mathbf w_{2} \ldots \mathbf w_{P} \right]\in \mathbb{C}^{M \times P}$ and $\mathbf S = \text{diag} \left(\left[s_{1}\ s_{2} \ldots s_{P}\right]\right) \in \mathbb{C}^{P \times P}$. For simplicity, we assume that $\mathbf S=\sqrt{P_{t}}\mathbf I_{P}$, where $P_{t}$ is the power of the pilot symbol. Then we have 
\begin{eqnarray} \label{Eq_Heff}
\mathbf Y=  \sqrt{P_{t}}\mathbf H \mathbf P  + \mathbf W = \sqrt{P_{t}}\mathbf B \mathbf \Sigma \mathbf A^{H} \mathbf P +\mathbf {W}.
\end{eqnarray} 
Our goal is to estimate the channel matrix $\mathbf H \in \mathbb{C}^{M \times N}$ from the measurements $\mathbf Y \in \mathbb{C}^{M \times P}$. Note that the number of pilots is usually smaller than the number of transmit antennas, i.e., $P<N$. Hence, we need to exploit the sparsity of $\mathbf H$ for its estimation, which will be discussed in the next section.

\subsection{Existing mmWave Channel Estimation Methods}
Before describing our proposed mmWave channel estimator, we briefly discuss some existing mmWave channel estimation methods 
\cite{Channel_Estimation_and_Hybrid_Precoding,Exploiting_spatial_sparsity_for_estimating_channels_of_hybrid_MIMO_systems_in_mmWave,Millimeter_Wave_MIMO_channel_estimation_Adaptive_CS,2D_MUSIC_mmWave} which can be divided into two categories.  

\subsubsection {CS-based mmWave channel estimators}
The mmWave channel is usually composed of small number of propagation paths and CS-based algorithms have been developed  \cite{Channel_Estimation_and_Hybrid_Precoding, Exploiting_spatial_sparsity_for_estimating_channels_of_hybrid_MIMO_systems_in_mmWave,Millimeter_Wave_MIMO_channel_estimation_Adaptive_CS}
for channel estimation.  First the dictionary matrices $\mathbf A_{D}\in \mathbb{C}^{N\times N_{G}}$ and $\mathbf B_{D} \in \mathbb{C}^{M\times N_{G}}$ are constructed based on quantized AoD angle of the transmitter and AoA angle of the receiver. The AoDs and AoAs are assumed to be taken from a uniform grid of $N_{G}$ points with $N_{G} \gg L$. The resulting dictionary matrix is expressed (take the transmitter $\mathbf A_{D}$ for example, the receiver dictionary matrix $\mathbf B_{D}$ is  similar.)
\begin{eqnarray}
\mathbf A_{D} = \left[ \mathbf a(\mathbf{\bar g}_{1}) \  \mathbf a(\mathbf{\bar g}_{2}) \ldots \mathbf a(\mathbf{\bar g}_{N_{G}}) \right], 
\end{eqnarray}
where $\mathbf{\bar g}_{i}= \left\{{\bar g}_{i,1},  {\bar g}_{i,2} \right\} = \left\{ \frac{1}{2}\sin\left(\bar
\vartheta_{i}\right)\cos(\bar \varphi_{i}), \frac{1}{2}\cos(\bar \vartheta_{i})  \right\}$ with $\bar \vartheta_{i}=\frac{\left( i-1 \right)2\pi}{N_{G}}-\pi$, $\bar \varphi_{i}=\frac{\left( i-1 \right)2\pi}{N_{G}}-\pi$ denotes the transmit array response vector for the grid point $\bar \vartheta_{i}$ and $\bar \varphi_{i}$ for $i=1, 2, \ldots, N_{G}$. The size $N_G$ of the angle grids can be set according to the desired angular resolution. On this basis, the received signal $\mathbf Y$ in \eqref{Eq_Heff} can be vectorized as \cite{Channel_Estimation_and_Hybrid_Precoding}
\begin{eqnarray}\label{vector_data_y}
\mathbf y &=& \text{vec}\left(\mathbf Y \right)  = \sqrt{P_{t}} \left(\mathbf P^{T} \otimes \mathbf
I_{M}\right) \text{vec}\left(\mathbf H \right) + \mathbf w  \\
&=& \sqrt{P_{t}}  {\left(\mathbf P^{T} \otimes \mathbf
I_{M}\right) \left( \mathbf A^{*}_{D} \odot \mathbf{B}_{D} \right)}\mathbf x  + \mathbf w = \sqrt{P_{t}}\mathbf G \mathbf x + \mathbf w, 
\end{eqnarray}
where $ \odot $ denotes the matrix Khatri-Rao products,  $\left( \cdot \right)^{T}$ denotes the  transpose operation, $\left( \cdot \right)^{*}$ denotes the complex conjugate, $\mathbf x \in \mathbb{C}^{N_{G}^{4}}$ is a sparse vector that has non-zero elements in the locations associated with the dominant paths. Note that the angle spaces of interest are discretized into a large number of grids, and the actual AoA and AoD angles may not exactly reside on the predefined grids. Those off-grid angles can lead to mismatches in the channel model and degrade the estimation performance.

\subsubsection{Subspace-based mmWave channel estimators}
Another existing approach to mmWave channel estimation is based on the subspace methods such as the MUSIC algorithm \cite{2D_MUSIC_mmWave}. The MUSIC algorithm  firstly calculates the covariance matrix  of the received signal $\mathbf Y$ and then finds the signal and noise subspaces via eigendecomposition. It then estimates each channel path's array response, i.e., $\mathbf {\hat g}_{l}$ and $\mathbf {\hat f}_{l}$ for $l=1,2,\ldots,\hat L$, where $\hat L$ is the estimated number of paths, by exploiting the orthogonality between the signal and noise subspaces. Finally, each channel path's coefficient, i.e., $\hat \sigma_{l}$ can be estimated via the least-squares (LS) method. The MUSIC algorithm has been popular for its good resolution and accuracy in AoD/AoA estimation \cite{MUSIC_DoA_Estimation}, \cite{Joint_Transmission_Reception_Diversity_Smoothing}. However, it is also reported that the off-grid CS method \cite{Compressed_Sensing_Off_the_Grid} can outperform the MUSIC algorithm in terms of estimation accuracy in noisy environments \cite{Atomic_Norm_Denoising,Exact_Joint_Sparse}.  

%\begin{figure}[htbp]
%        \centering
%        \includegraphics[width=0.9\linewidth]{figures/TransFigure1.pdf}
%        \caption{A beam training stage: training beam is transmitted from training from symbol $1$ to symbol $P$ and each training symbol is associated with a specific beam $p$} \label{fig1}
%\end{figure}

\section{Channel Estimation via Atomic Norm Minimization\label{Section_chest}}
As explained in the previous section, the performance of the  mmWave channel estimators based on on-grid methods such as CS can be degraded due to grid mismatch. In this section, we propose a new mmWave channel estimator based on an off-grid CS method, i.e., the atomic norm minimization method. 
\subsection{Background on Multi-dimensional Atomic Norm}
First we briefly introduce the concept of multi-dimensional atomic norm \cite{Vandermonde_Decomposition_Multilevel_Toeplitz_Matrices}. A $d$-dimensional ($d$-dim) atom is defined as  
\begin{eqnarray} \label{d_dim_atom}
\mathbf q_{d}\left( x_{1}, \ldots, x_{d} \right) = \mathbf c_{n_{1}} \left( x_{1} \right) \otimes \ldots \otimes \mathbf c_{n_{d}} \left( x_{d} \right), 
\end{eqnarray}
where $n_{i} $ is the length of the normalized vector $\mathbf c_{n_{i}}\left(x_{i}\right)$ defined in \eqref{Eq_uniform_sin} and $x_{i} \in \left[-\frac{1}{2}, \frac{1}{2} \right)$ for $i=1, 2,\ldots, d$.
The $d$-dim atomic set is then given by 
\begin{eqnarray} \label{d_dim_atom_set}
\mathcal{A}= \left\{ \mathbf q_{d}\left( x_{1}, \ldots, x_{d} \right) : x_{i} \in \left[-\frac{1}{2}, \frac{1}{2} \right),\ i=1,\ldots, d  \right\}.
\end{eqnarray}
For any vector $\mathbf t_{d}$  of the form $\mathbf t_{d} = \sum\limits_{l} { {\alpha_{l} \mathbf q_{d}}(x_{l,1}, x_{l,2}, \ldots, x_{l,d})}$, its $d$-dim atomic norm with respect to $\mathcal{A}$ is defined as 
\begin{eqnarray} \label{d_dim_atomic_norm}
{\left\| \mathbf t_{d} \right\|_{\cal A}}  &=& \mathop{\inf} \left\{ t :
\mathbf t_{d} \in t {\text {conv}} \left(\cal A \right) \right\}, \nonumber \\
&=&  \mathop {\inf }\limits_{\scriptstyle{ x_{l,1}, x_{l,2}, \ldots, x_{l,d}} \in [-\frac{1}{2},\frac{1}{2}) \hfill\atop
\scriptstyle{ \alpha_{l}} \in \mathbb{C} \hfill} \left\{
{\left. {\sum\limits_{l} \left\vert \alpha_{l} \right\vert} \right| \mathbf t_{d} = \sum\limits_{l} { {\alpha_{l} \mathbf q_{d}}(x_{l,1}, x_{l,2}, \ldots, x_{l,d})}} \right\},
\end{eqnarray}
where  ${\text {conv}} \left(\cal A \right) $ is the convex hull of $\cal A$. The  $d$-dim atomic norm of $\mathbf t_{d}$ has following
equivalent form \cite{Vandermonde_Decomposition_Multilevel_Toeplitz_Matrices}:
\begin{eqnarray} \label{anm_HD}
\left\Vert \mathbf t_{d}\right\Vert_{\mathcal{A}} = \inf_{ \mathcal{U}_{d} \in \mathbb{C}^{(2n_{d}-1) \times (2n_{d-1}-1) \times \ldots \times (2 n_{1}-1)} , t \in \mathbb{R}}  
\begin{Bmatrix}
\frac{1}{2n_{1}n_{2}\ldots n_{d}}\text{Tr}\left(  \mathbb T_d(\mathcal U_{d}) \right) + \frac{1}{2}
t        \\
\\ 
\text{s.t.}\left[ {\begin{array}{*{20}{c}}
        \mathbb T_d(\mathcal U_{d}) & \mathbf t_{d} \\
        \mathbf t_{d}^{H} & t
        \end{array}} \right] \succeq 0 
\end{Bmatrix}, 
\end{eqnarray}
where $\text{Tr}\left(\cdot\right)$ is the trace of the input matrix, $\mathcal U_{d} \in \mathbb{C}^{(2n_{d}-1) \times (2n_{d-1}-1) \times \ldots \times (2 n_{1}-1)}$ is a $d$-way tensor and $\mathbb T_d(\mathcal U_d)$ is a $d$-level block Toeplitz, which is defined recursively as follows. Denote $\mathbf
n_d = (n_{d}, n_{d-1}, \ldots, n_{1})$ and $\mathcal U_{d-1}(i) = \mathcal U_{d}(i,:,:,...,:)$ for $i = -n_{d}+1,-n_{d}+2,...,n_{d}-1$. For $d=1$, $\mathbf{n}_{1}=\left(n_{1}\right)$ and
$\mathbb T_1(\mathbf u_1) = \text{Toep}(\mathbf u_1)$ with $\mathbf u_1 \in \mathbb{C}^{(2 n_1
- 1) \times 1}$, i.e., 
\begin{eqnarray}
\label{eq:T}
\mathbb T_1(\mathbf u_1) = \text{Toep}(\mathbf u_1) = \left[ {\begin{array}{*{20}{c}}
        {u_{1}(0)}&{{u_{1}(1)}}&{\cdots}&{{u_{1}(n_1-1)}}\\
        {u_{1}(-1)}&{u_{1}(0)}&{\cdots}&{u_{1}(n_1-2)}\\
        {\vdots}&{\vdots}&{\ddots}&{\vdots}\\
        {u_{1}(1-n_1)}&{u_{1}(2 - n_1)}&{\cdots}&{u_{1}(0)}
        \end{array}} \right].
\end{eqnarray}
For $d \geq 2$, we have
\begin{eqnarray}  \label{eq:Toeplitz2}
\mathbb T_d(\mathcal {U}_d) = \left[ {\begin{array}{*{20}{c}}
        {\mathbb T_{d-1}(\mathcal U_{d-1}(0))}&{\mathbb T_{d-1}(\mathcal
U_{d-1}(1))}&{\ldots}&{\mathbb T_{d-1}(\mathcal U_{d-1}(n_{d}-1))}\\
        {\mathbb T_{d-1}(\mathcal U_{d-1}(-1))}&{\mathbb T_{d-1}(\mathcal
U_{d-1}(0))}&{\ldots}&{\mathbb T_{d-1}(\mathcal U_{d-1}(n_{d}-2))}\\
        {\vdots}&{\vdots}&{\ddots}&{\vdots}\\
        {\mathbb T_{d-1}(\mathcal U_{d-1}(1-n_{d}))}&{\mathbb T_{d-1}(\mathcal
U_{d-1}(2-n_{d}))}&{\ldots}&{\mathbb T_{d-1}(\mathcal U_{d-1}(0))}
        \end{array}} \right].
\end{eqnarray}

\subsection{Atomic Norm Minimization Formulation\label{2D_atomic_norm_chest}}

In this subsection, we formulate the atomic norm minimization problem for channel estimation. First, we vectorize the mmWave FD-MIMO channel matrix $\mathbf H$ in \eqref{Eq_2_2} as
\begin{eqnarray} \label{Eq_3B_1}
\mathbf h &=& \text{vec}(\mathbf H) = \sum_{l=1}^{L} \sigma_{l} \mathbf a(\mathbf g_{l})^{*} \otimes \mathbf b(\mathbf f_{l}) \nonumber \\
&=& \sum_{l=1}^{L} \sigma_{l} \Big( \mathbf c_{N_{1}}\left( g_{l,1} \right) \otimes \mathbf c_{N_{2}} \left( g_{l,2} \right) \Big)^{*} \otimes \Big(\mathbf c_{M_{1}} \left( f_{l,1} \right) \otimes \mathbf c_{M_{2}} \left( f_{l,2} \right)\Big) \nonumber \\
&=& \sum_{l=1}^{L} \sigma_{l} \mathbf c_{{N_1}}^*({g_{l,1}})
\otimes \mathbf c_{{N_2}}^*({g_{l,2}}) \otimes {\mathbf c_{{M_1}}}({f_{l,1}})
\otimes {\mathbf c_{{M_2}}}({f_{l,2}}).  
\end{eqnarray} 
Comparing \eqref{d_dim_atomic_norm} and \eqref{Eq_3B_1}, for the mmWave FD-MIMO channel with UPA  configuration, the atom has the form of
\begin{eqnarray} \label{atom4D_definition}
\mathbf q_{4}\left( \mathbf g, \mathbf f \right) =  \mathbf c_{{N_1}}^*({g}_{1}) \otimes \mathbf c^{*}_{N_{2}}({g}_{2}) \otimes {\mathbf c_{{M_1}}}({f}_{1}) \otimes {\mathbf c_{{M_2}}}({f}_{2}), 
\end{eqnarray}
and the  set of atoms is defined as the collection of all normalized 4D complex sinusoids: $ \mathcal{A} = \left\{\mathbf q_{4}\left( \mathbf g, \mathbf f \right) : \mathbf f \in [-\frac{1}{2},\ \frac{1}{2}) \times [-\frac{1}{2},\ \frac{1}{2}),\ \mathbf g \in [-\frac{1}{2},\ \frac{1}{2}) \times [-\frac{1}{2},\ \frac{1}{2}) \right\}$ \cite{Compressive_Recovery_2D_off_grid,ChiC15}. The 4D atomic norm  for any $\mathbf h$ defined in \eqref{Eq_3B_1} can be written as \cite{Compressive_Recovery_2D_off_grid}:
\begin{eqnarray}
\label{eq:atomnorm_2D}
\left\| \mathbf h \right\|_{\cal A} = \inf_{\substack{
\mathbf f_{l} \in [-\frac{1}{2},\ \frac{1}{2}) \times [-\frac{1}{2},\ \frac{1}{2}),\\ \mathbf g_{l} \in [-\frac{1}{2},\ \frac{1}{2}) \times [-\frac{1}{2},\ \frac{1}{2}),\\ \sigma_{l} \in \mathbb{C}
  }} \left\{ \left. \sum\limits_{l} {\left\vert\sigma_{l}\right\vert}
\right| \mathbf h = \sum\limits_{l} \sigma_{l} \mathbf q_{4} \left( \mathbf g_{l}, \mathbf f_{l} \right)  \right\}.  
\end{eqnarray}
The atomic norm can enforce sparsity in the atom set $\mathcal{A}$. On this basis, an optimization problem will be formulated for the estimation of the path frequencies  $\left\{\mathbf f_{l}, \mathbf g_{l}\right\}$. For the convenience of calculation, we will use the equivalent form of the atomic norm given by \eqref{anm_HD}, i.e.,
\begin{eqnarray} \label{anm_2D}
\left\Vert \mathbf h\right\Vert_{\mathcal{A}} = \inf_{\substack{\mathcal U_{4} \in \mathbb{C}^{(2N_{1}-1) \times (2N_{2}-1) \times (2M_{1}-1) \times (2M_{2}-1)}, \\t \in \mathbb{R}}}  
\begin{Bmatrix}
\frac{1}{2MN}\text{Tr}\left(  \mathbb T_4(\mathcal U_{4}) \right) + \frac{1}{2}
t        \\
\\ 
\text{s.t.}\left[ {\begin{array}{*{20}{c}}
        \mathbb T_4(\mathcal U_{4}) & \mathbf h \\
        \mathbf h^{H} & t
        \end{array}} \right] \succeq 0 
\end{Bmatrix}, 
\end{eqnarray}
where $\mathbb T_{4}(\mathcal U_{4})$ is a $4$-level Toeplitz matrix defined in \eqref{eq:Toeplitz2}. Define the minimum frequency separations as
\begin{eqnarray}
\Delta_{\min,f_{i}} &=& \min_{l \neq l'} \min \{ |f_{l,i}-f_{l',i}|, 1 - |f_{l,i}-f_{l',i}| \}, \\
\Delta_{\min,g_{i}} &=& \min_{l \neq l'} \min \{ |g_{l,i}-g_{l',i}|, 1 - |g_{l,i}-g_{l',i}| \},
\end{eqnarray}
for $i = 1,2$. To show the connection between the atomic norm and the channel matrix, we obtain the following theorem via extending Theorem 1.2 in \cite{candes2014towards} for 1D atomic norm to 4D atomic norm.  
\begin{theorem} \label{Theorem1}
   If the path component frequencies are sufficiently separated, i.e.,
        \begin{eqnarray}\label{freq_f_separation}
        \Delta_{\min,f_{i}} &\geq& \frac{1}{\left\lfloor(M_{i}-1)/4\right\rfloor},\\ \label{freq_g_separation}
        \Delta_{\min,g_{i}} &\geq& \frac{1}{\left\lfloor(N_{i}-1)/4\right\rfloor},
        \end{eqnarray}
        for $i=1,2$, then we have $\|\mathbf h \|_{{\cal A}} = \sum_{l} \left\vert\sigma_l\right\vert$, so the component atoms of $\mathbf h$ can be uniquely located via computing its atomic norm.
\end{theorem}
The proof follows the same line as that in Theorem 1.2 \cite{candes2014towards}, with the dual polynomial constructed by interpolation with a 4D kernel. The theorem holds because all bounds in the proof of [Theorem 1.2, 34] hold by leveraging the 1D results.

To estimate the mmWave FD-MIMO channel $\mathbf H$ in \eqref{Eq_2_2} based on the signal $\mathbf Y$ in \eqref{Eq_Heff}, we then formulate the following optimization problem:
\begin{eqnarray} \label{Eq_2D_atomic_norm_opt}
\mathbf{\hat h} = \min_{\mathbf h \in \mathbb{C}^{MN}} \mu {\left\| \mathbf h \right\|_{\cal A}}\ + \frac{1}{2} \left\Vert\mathbf{y} - \sqrt{P_t} \left(\mathbf P^{T} \otimes \mathbf I_{M}\right) \mathbf h \right\Vert^{2}_{2},
\end{eqnarray}
where $\mathbf{y} = \text{vec}(\mathbf {Y})$ is given by \eqref{vector_data_y} and $\mu \propto \sigma_{w}\sqrt{MN\log\left( MN \right)} $ is a weight factor \cite{Super-Resolution-Radar}. Using \eqref{anm_2D}, \eqref{Eq_2D_atomic_norm_opt} can be equivalently formulated as a semi-definite program (SDP): 
\begin{eqnarray} \label{Eq_2D_atomic_norm_sdp}
\min_{\substack{ \mathcal U_{4} \in
\mathbb{C}^{(2N_{1}-1) \times (2N_{2}-1) \times (2M_{1}-1) \times (2M_{2}-1)},\\\mathbf h \in \mathbb{C}^{MN},\ t \in \mathbb{R}}} && \frac{\mu}{2MN}\text{Tr}
\left( \mathbb T_4(\mathcal U_{4}) \right) + \frac{\mu}{2} t + \frac{1}{2} \left\Vert\mathbf{y} - \sqrt{P_t} \left(\mathbf P^{T} \otimes \mathbf I_{M}\right) \mathbf h \right\Vert^{2}_{2}\
\nonumber \\ \text{s.t} &&  \left[ {\begin{array}{*{20}{c}}
{\mathbb T_4(\mathcal U_{4}) }& \mathbf h \\
{\mathbf h^{H}}& {t}
\end{array}} \right] \succeq\ 0. 
\end{eqnarray}
The above problem is convex, and can be solved by using a standard convex solver. Suppose the solution to \eqref{Eq_2D_atomic_norm_sdp} is $\mathbf{\hat h}$. Then the estimated channel matrix is given by $\mathbf{\hat H} = {\rm vec}^{-1}\left( \mathbf{\hat h} \right)$ where ${\rm vec}^{-1}(\cdot)$ is the inverse operation of ${\rm vec}(\cdot)$.

\section{Efficient Algorithm for Channel Estimation under UPA}\label{SEC_AN_APPROX}
\subsection{A Formulation Based on 2D MMV Atomic Norm}

Note that the dimension of the positive semidefinite matrix in \eqref{Eq_2D_atomic_norm_sdp} is $(MN+1) \times (MN+1)$, and the 4D atomic norm minimization formulation is of high computational complexity and has large memory requirements. To reduce the complexity, we can treat $\mathbf Y$ as 2D multiple measurement vectors (MMV) \cite{Exact_Joint_Sparse} in transmit and receive dimensions.

Unlike the 4D atomic norm that is calculated with input vector $\mathbf h$, the MMV atomic norm is calculated with the matrix input $\mathbf H$. Specifically, we define the atom $\mathbf{\bar Q} \left(\mathbf f, \mathbf{\bar a} \right)
= \mathbf b\left(\mathbf f \right) \mathbf{\bar a}^{H}$ with $\mathbf f \in
[-\frac{1}{2},\ \frac{1}{2}) \times [-\frac{1}{2},\ \frac{1}{2})$, and $\mathbf{\bar
        a} \in \mathbb{C}^{N \times 1}$ with $\left\Vert \mathbf{\bar a} \right\Vert_{2}=1$. Correspondingly, the atom set is defined as
\begin{eqnarray}
\label{eq:1DMMV_atoms_set_def}
{\cal{A}_{\text{MMV}}} = \left\{ \mathbf{\bar Q}\left( \mathbf f, \mathbf{\bar a}
\right) : \mathbf f \in [-\frac{1}{2},\ \frac{1}{2}) \times [-\frac{1}{2},\
\frac{1}{2}), \left\Vert \mathbf{\bar a} \right\Vert_{2}=1 \right\}.
\end{eqnarray}
It is worth noting that $\mathbf{\bar a}$ is not restricted by the structural constraint in \eqref{Eq_2_3}. With \eqref{eq:1DMMV_atoms_set_def}, we extend the 1D MMV atomic norm \cite{Exact_Joint_Sparse} to the 2D MMV atomic norm of $\mathbf H$ defined by
\begin{eqnarray} \label{eq:atomic_norm_def}
{\left\| \mathbf H \right\|_{\cal A_{\text{MMV}}}}= \mathop {\inf }_{\substack{
                \mathbf f_{l} \in [-\frac{1}{2},\ \frac{1}{2}) \times [-\frac{1}{2},\ \frac{1}{2}),\\ \bar{\mathbf a}_l \in \mathbb{C}^{N \times 1},\ \sigma_{l} \in \mathbb{C}
        }} \left\{ {\left.
        {\sum\limits_{l} |\sigma_l| } \right|
        \mathbf
        H = \sum\limits_{l} \sigma_l \mathbf{\bar Q}\left( \mathbf f_{l},
        \mathbf{\bar a}_l  \right), \left\Vert \mathbf{\bar a} \right\Vert_{2}=1 } \right\}.
\end{eqnarray}
This atomic norm is equivalent to the solution of the following SDP \cite{Exact_Joint_Sparse}:
\begin{eqnarray}
\label{eq:AN_SDP_1D}
\left\Vert \mathbf H\right\Vert_{\mathcal{A}_{\text{MMV}}} = \inf_{\mathbf
        U_{2} \in \mathbb{C}^{(2M_2-1) \times (2M_1-1)},
        \mathbf X \in \mathbb{C}^{N\times N}}  
\begin{Bmatrix}
\frac{1}{2M}\text{Tr} \left( \mathbb T_2(\mathbf U_{2}) \right) + \frac{1}{2N}\text{Tr}
\left( \mathbf X\right)     \\
\\ 
\text{s.t.} \left[ {\begin{array}{*{20}{c}}
        \mathbb T_2(\mathbf U_{2}) & \mathbf H \\
        \mathbf H^{H} & \mathbf X
        \end{array}} \right] \succeq 0 
\end{Bmatrix},
\end{eqnarray}
where $\mathbf X$ is constrained to be a Hermitian matrix. Then using \eqref{Eq_Heff}, we can
formulate the following optimization problem for channel estimation:
\begin{eqnarray}\label{2D_MMV_AN}
\mathbf{\hat H} = \min_{\mathbf H \in \mathbb{C}^{M \times N}} \mu {\left\| \mathbf H \right\|_{{\cal
                        A}_{\text{MMV}}}}\ + \frac{1}{2} \left\Vert \sqrt{P_{t}} \mathbf H\mathbf P -\mathbf{Y}  \right\Vert^{2}_{F},
\end{eqnarray}
where $\left\Vert \cdot\right\Vert_F$ denotes matrix Frobenius norm. Plugging \eqref{eq:AN_SDP_1D} into \eqref{2D_MMV_AN}, the size of the positive semidefinite matrix in the constraint is $\left(M+N\right)\times \left(M+N\right)$, resulting in considerably lower computational complexity and memory requirement than \eqref{Eq_2D_atomic_norm_sdp}.

\subsection{An Approximation to 4D Atomic Norm Minimization \label{lrmc_chest}}

Next we propose an approximation to the 4D atomic norm to reduce the computational complexity. In \cite{2D_AN_approx}, the authors explore the approximation of 2D atomic norm to improve the efficiency. Here, we extend the results from 2D atomic norm to 4D atomic norm case. Similar to the 2D MMV atomic norm, the proposed approximation is calculated with input $\mathbf H$. From \eqref{Eq_2_2}, $\mathbf H$ is the sum of $\sigma_{l} \mathbf b(\mathbf f_{l} )  \mathbf a(\mathbf g_{l} )^{H}$, in which both $\mathbf a(\mathbf g_{l})$ and $\mathbf b(\mathbf f_{l})$ are Fourier bases. Different from the vectorized atomic norm, we introduce the matrix atom $\mathbf Q\left( \mathbf f, \mathbf g \right) = \mathbf b(\mathbf f )  \mathbf a(\mathbf g )^{H}$ and the matrix atom set
\begin{eqnarray}
\label{eq:atoms_matrix_set_def}
{\mathcal{A}_{M}} = \left\{ \mathbf{Q}\left( \mathbf f, \mathbf g
\right) = \mathbf{b}(\mathbf f) \mathbf{a}(\mathbf g)^H : \mathbf f \in [-\frac{1}{2},\ \frac{1}{2}) \times [-\frac{1}{2},\
\frac{1}{2}), \mathbf g \in [-\frac{1}{2},\ \frac{1}{2}) \times [-\frac{1}{2},\
\frac{1}{2}) \right\}.
\end{eqnarray}
The matrix atomic norm is then given by
\begin{eqnarray}
\label{eq:atomic_matrix_norm_2D}
\left\| \mathbf H \right\|_{\mathcal A_{M}} = \inf_{\substack{
\mathbf f_{l} \in [-\frac{1}{2},\ \frac{1}{2}) \times [-\frac{1}{2},\ \frac{1}{2}),\\
\mathbf g_{l} \in [-\frac{1}{2},\ \frac{1}{2}) \times [-\frac{1}{2},\ \frac{1}{2}),\\
\sigma_{l} \in \mathbb{C}
  }} \left\{ \left. \sum\limits_{l} {\left\vert\sigma_{l}\right\vert}
\right| \mathbf H = \sum\limits_{l} \sigma_{l} \mathbf Q \left( \mathbf
f_{l}, \mathbf g_{l} \right)  \right\}.  
\end{eqnarray}

The matrix atom set is composed of rank-one matrices, and hence it amounts to atomic norm of low rank matrices. Since the operator $\text{vec}(\cdot)$ is a one-to-one mapping and the mapping ${\cal A}_M  \to {\cal A}$ is also one-to-one, it is straightforward to conclude that $\| \mathbf H \|_{{\cal A}_M} = \| \mathbf h \|_{\cal A}$. Hence, if the component frequencies satisfy the sufficient separation condition given by \eqref{freq_f_separation} and \eqref{freq_g_separation}, we have $\| \mathbf H \|_{{\cal A}_M} = \sum_{l} \left\vert\sigma_l\right\vert$ by Theorem \ref{Theorem1}.

Finding the harmonic components via atomic norm is an infinite programming problem over all feasible $\mathbf f$ and $\mathbf g$, which is difficult. For better efficiency, we use ${\rm SDP}(\mathbf H)$ in the following Lemma to approximate $\|\mathbf H \|_{{\cal A}_M}$.

\begin{lemma} \label{Lemma1}
        For $\mathbf H$ given by \eqref{Eq_2_2}, we have $\left\Vert \mathbf H\right\Vert_{\mathcal{A}_M} \geq \rm{SDP}(\mathbf H) \geq \left\Vert \mathbf H\right\Vert_{\mathcal{A}_{\rm{MMV}}}$,
        where
\begin{eqnarray} \label{eq:SDP_H}
 {\rm SDP}(\mathbf H) \triangleq \inf_{\mathbf U_{2} \in \mathbb{C}^{(2M_2-1) \times (2M_1-1)},
        \mathbf V_{2} \in \mathbb{C}^{(2N_2-1)\times (2N_1-1)}}  
\begin{Bmatrix}
\frac{1}{2M}{\rm Tr}\left( \mathbb T_2(\mathbf U_{2}) \right) + \frac{1}{2N}\rm{Tr}
\left(\mathbb T_2(\mathbf V_{2}) \right)        \\
\\ 
\text{s.t.}\ \left[ {\begin{array}{*{20}{c}}
        \mathbb T_2(\mathbf U_{2}) & \mathbf H \\
        \mathbf H^{H} & \mathbb T_2(\mathbf V_{2})
        \end{array}} \right] \succeq 0 
\end{Bmatrix},
\end{eqnarray}
with $\mathbb T_{2}(\mathbf U_{2})$ and $\mathbb T_{2}(\mathbf V_{2})$ being $2$-level Toeplitz matrices defined in \eqref{eq:Toeplitz2}. 
\end{lemma}
\begin{proof}
The relation ${\rm SDP}\left( \mathbf H \right) \geq \left\Vert \mathbf H\right\Vert_{\mathcal{A}_{\text{MMV}}}$ can be directly obtained from the definitions in \eqref{eq:AN_SDP_1D} and \eqref{eq:SDP_H}. It remains to show $\left\Vert \mathbf H\right\Vert_{\mathcal{A}_M} \geq \rm{SDP}(\mathbf H)$. Denote
\begin{eqnarray}
\mathbf{\tilde a}(\mathbf g_{l}, \omega_{l}) &=& \frac{1}{\sqrt{N}}e^{j2
                \pi \omega_{l}} \mathbf c_{{N_1}}^*({g_{l,1}}) \otimes
            \mathbf c_{{N_2}}^*({g_{l,2}}),
        \nonumber \\
\mathbf{\tilde b}(\mathbf f_{l}, \chi_{l}) &=& \frac{1}{\sqrt{M}}e^{j2
                \pi \chi_{l}} \mathbf c_{{M_1}}({f_{l,1}}) \otimes \mathbf
        c_{{M_2}}({f_{l,2}}), \nonumber
        \end{eqnarray}
with $\omega_{l} \in [0, 2\pi]$ and $\chi_{l} \in [0, 2\pi]$ such that $\sigma_{l} = \left\vert\sigma_{l}  \right\vert e^{j2\pi \left(\omega_{l}+\chi_{l}\right)}$. For any $\mathbf H = \sum_{l} \sigma_{l} \mathbf{b}\left( \mathbf f_{l} \right)\mathbf{a}\left( \mathbf g_{l} \right)^{H}$, if we set 
\begin{eqnarray}
\mathbf U_2 &=& [\mathbf u_1(-M_1+1),\mathbf u_1(-M_1+2),...,\mathbf u_1(M_1-1)],\\
\mathbf V_2 &=& [\mathbf v_1(-N_1+1),\mathbf v_1(-N_1+2),...,\mathbf v_1(N_1-1)],
\end{eqnarray}
where
\begin{eqnarray}
\mathbf u_1(i) = \frac{1}{\sqrt M} \sum_l |\sigma_l| \mathbf{\tilde c}_{M_2}(f_{l,2}) e^{j 2 \pi (i-1) f_{l,1} }, \\
\mathbf v_1(i) = \frac{1}{\sqrt N} \sum_l |\sigma_l| \mathbf{\tilde c}_{N_2}^*(g_{l,2}) e^{-j 2 \pi (i-1) g_{l,1} },
\end{eqnarray}
with $\mathbf{\tilde c}_{n} (x) = \frac{1}{\sqrt{n}} \left[ e^{j2\pi (1-n) x}, e^{j2\pi (2-n) x}, \cdots, e^{j2\pi \left( n-1 \right) x } \right]^{T} \in \mathbb{C}^{2n \times 1}$, then the 2-level Toeplitz matrices $\mathbb T_2(\mathbf U_2)$ and $\mathbb T_2(\mathbf V_2)$ satisfy
\begin{eqnarray} \label{Eq_P1_5}
\mathbb T_2(\mathbf U_2) &=& \sum_{l} | \sigma_l | \mathbf{b}\left(\mathbf
f_{l} \right)\mathbf{b}\left(\mathbf f_{l} \right)^{H} \nonumber \\
&=& \sum_{l} | \sigma_l | \mathbf {\tilde b}\left(\mathbf f_{l},
\chi_{l} \right)\mathbf {\tilde b}\left(\mathbf f_{l}, \chi_{l}
\right)^{H},
\end{eqnarray}
\begin{eqnarray} \label{Eq_P1_6}
\mathbb T_2(\mathbf V_2) &=& \sum_{l} | \sigma_l | \mathbf {a}\left(\mathbf
g_{l} \right)\mathbf{a}\left( \mathbf g_{l} \right)^{H} \nonumber \\
&=& \sum_{l} | \sigma_l | \mathbf {\tilde a}\left(\mathbf g_{l},
\omega_{l} \right)\mathbf {\tilde a}\left(\mathbf g_{l}, \omega_{l}
\right)^{H}.
\end{eqnarray} 
Moreover, the matrix
        \begin{eqnarray} \label{Eq_P1_7}
        \mathbf M=\left[ {\begin{array}{*{20}{c}}
                \mathbb T_2(\mathbf U_2) & \mathbf H \\
                \mathbf H^{H} & \mathbb T_2(\mathbf V_2)
                \end{array}} \right] = \sum_{l} | \sigma_l | \begin{bmatrix}
        \mathbf {\tilde b}\left(\mathbf f_{l},\chi_{l} \right) 
        \\ \mathbf {\tilde a} \left(\mathbf g_{l},
        \omega_{l} \right)   \\ \end{bmatrix} 
        \begin{bmatrix}  \mathbf {\tilde b}\left(\mathbf f_{l}, \chi_{l}
        \right) 
        \\ \mathbf {\tilde a}\left(\mathbf g_{l}, \omega_{l} \right) 
        \\\end{bmatrix}^{H}
        \end{eqnarray}
        is positive semidefinite, indicating that the constraints in \eqref{eq:SDP_H} are satisfied. Note that $ {\rm SDP}(\mathbf H) \leq \frac{1}{2M}{\rm Tr} \left( \mathbb T_2(\mathbf U_2) \right) + \frac{1}{2N}{\rm Tr} \left( \mathbb T_2(\mathbf V_2) \right) = \sum_{l} | \sigma_l |$ according to the definition in \eqref{eq:SDP_H}. Since this holds for any decomposition of $\mathbf H$, we obtain ${\rm SDP}\left( \mathbf H \right) \leq \left\Vert \mathbf H\right\Vert_{\mathcal{A}_{M}}$. 
\end{proof}

The above lemma shows that ${\rm SDP}(\mathbf H)$ is a lower bound of the matrix atomic norm. Moreover, the following lemma states that if the component frequencies are sufficiently separated, then ${\rm SDP}(\mathbf H)$ is equivalent to $\left\Vert \mathbf H\right\Vert_{\mathcal{A}_{M}}$.
\begin{lemma} \label{Lemma2}
If \eqref{freq_f_separation}-\eqref{freq_g_separation} hold, then $\left\Vert \mathbf H\right\Vert_{\mathcal{A}_{M}} = \rm{SDP}(\mathbf H)$.
\end{lemma}
\begin{proof}
First it follows from Theorem 4 in \cite{Exact_Joint_Sparse} that if \eqref{freq_f_separation}-\eqref{freq_g_separation} hold, then we have $\|\mathbf H \|_{{\cal A}_{\rm MMV}} = \sum_{l} \left\vert\sigma_l\right\vert$. Using Theorem 1 and the fact that $\|\mathbf h \|_{\cal A} = \|\mathbf H \|_{{\cal A}_M} $, we have $\|\mathbf H \|_{{\cal A}_M} = \|\mathbf H \|_{{\cal A}_{\text{MMV}}}$. Finally by Lemma \ref{Lemma1} we have $\|\mathbf H \|_{{\cal A}_M} = \|\mathbf H \|_{{\cal A}_{\text{MMV}}} = {\rm SDP}(\mathbf H)$.  
\end{proof}

When the sufficient separation condition given by \eqref{freq_f_separation} and \eqref{freq_g_separation}  is not satisfied, $\text{SDP}(\mathbf H)$ may not be the same as $\left\Vert \mathbf H\right\Vert_{\mathcal{A}_{M}}$. However, it is found via simulations that $\text{SDP}(\mathbf
H)$ still provides a good approximation to $\left\Vert \mathbf H\right\Vert_{\mathcal{A}_{M}}$ and usually results in good performance in channel estimation. Moreover, as  shown by Lemma \ref{Lemma1}, $\text{SDP}(\mathbf
H)$ is a lower bound of the atomic norm $\left\Vert \mathbf H\right\Vert_{\mathcal{A}_{M}}$ (or $\left\Vert \mathbf h\right\Vert_{\mathcal{A}}$ equivalently), i.e.,  $\left\Vert \mathbf h\right\Vert_{\mathcal{A}} = \left\Vert \mathbf H\right\Vert_{\mathcal{A}_M} \geq \text{SDP}(\mathbf H)$ in general.

\begin{figure*} 
        \centering
        \includegraphics[width=0.5\linewidth]{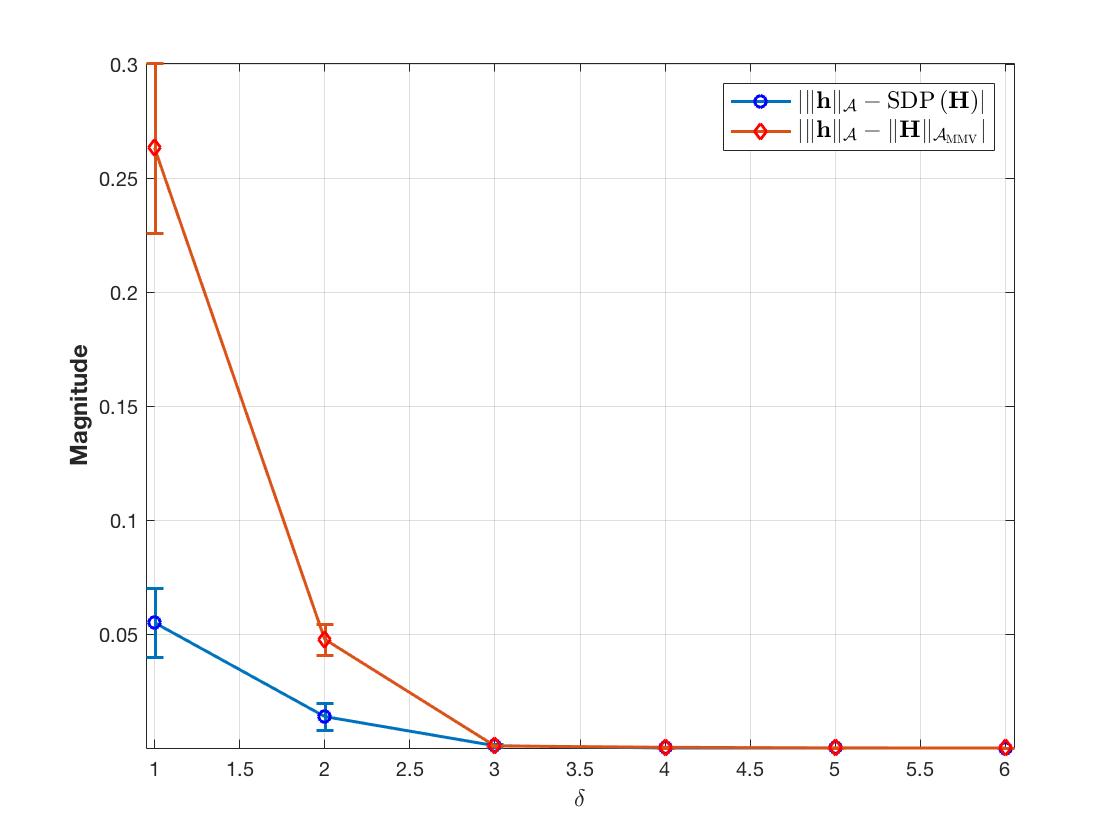}
        \caption{The approximation errors $\left\vert \left\Vert \mathbf H\right\Vert_{\mathcal{A}_M} - \text{SDP}(\mathbf H) \right\vert$ and $\left\vert \left\Vert \mathbf H\right\Vert_{\mathcal{A}_{M}} - \left\Vert \mathbf H\right\Vert_{\mathcal{A}_{\rm MMV}} \right\vert$ when the separations satisfy $\Delta_{\min,f_{i}} \geq \delta(M_i - 1)$,  $\Delta_{\min,g_{i}} \geq \delta/(N_i - 1)$, $N_{i} = M_{i}=16$, for $i=1,2$. The simulations are run 100 times for each $\delta$.} 
        \label{fig:Fig1}
\end{figure*}  

To show the approximation performances of both $\| \mathbf H \|_{{\cal A}_{\rm MMV}}$ and ${\rm SDP}(\mathbf H)$ to $\| \mathbf h \|_{\cal A}$, we perform a series of Monte Carlo trials for parameters $M_{1}=M_{2}=16$, $N_{1}=N_{2}=16$ with  $L=2$. $\mathbf f_{l}$ and  $\mathbf g_{l}$ take random values from $[-\frac{1}{2},\ \frac{1}{2}) \times [-\frac{1}{2},\ \frac{1}{2})$ such that the separations satisfy $\Delta_{\min,f_{i}} \geq \delta/\left\lfloor(N_i - 1)\right\rfloor$, $\Delta_{\min,g_{i}} \geq \delta / \left\lfloor(N_i - 1)\right\rfloor$ with $1 \leq \delta \leq 6$. In Fig. \ref{fig:Fig1}, we plot the approximation error against $\delta$ and the bars show $95\%$ confidence interval. As $\delta$ decreases, both approximation errors become larger. However, ${\rm SDP}(\mathbf H)$ provides a more accurate approximation than $\| \mathbf H \|_{{\cal A}_{\rm
MMV}}$. When $\delta \geq\ 4$, both approximation errors become zero.

Therefore, instead of solving the original 4D atomic norm minimization in \eqref{Eq_2D_atomic_norm_sdp}, we can solve the following SDP  
\begin{eqnarray} \label{Eq_3A_2}
\mathbf {\hat H} &=& \min_{\substack{ \mathbf H \in \mathbb{C}^{M \times N}, \\ \mathbf U_{2} \in \mathbb{C}^{(2M_2-1) \times (2M_1-1)}, \\
        \mathbf V_{2} \in \mathbb{C}^{(2N_2-1)\times (2N_1-1)} } } \frac{\mu}{2M}{\rm Tr}
\left( \mathbb T_2(\mathbf U_{2})  \right) + \frac{\mu}{2N}\rm{Tr} \left(
\mathbb T_2(\mathbf V_{2})  \right) + \frac{1}{2} \left\Vert \sqrt{P_{t}} \mathbf H \mathbf P - \mathbf{Y}\right\Vert^{2}_{F} \\ 
\text{s.t.} && \mathbf M = \left[ {\begin{array}{*{20}{c}}
\mathbb T_2(\mathbf U_{2}) & \mathbf H \\
\mathbf H^{H} & \mathbb T_2(\mathbf V_{2})
\end{array}} \right] \succeq 0. \nonumber
\end{eqnarray} 
The size of the positive semidefinite matrix in the constraint is $\left(M+N\right)\times \left(M+N\right)$, resulting in considerably lower computational complexity and memory requirement than \eqref{Eq_2D_atomic_norm_sdp}. 

\subsection{ADMM for Approximate 4D Atomic Norm Minimization}

To meet the requirement of real-time signal processing, we next derive an iterative algorithm for solving the SDP in \eqref{Eq_3A_2}, based on the alternating direction method of multipliers (ADMM) \cite{boyd2011distributed}. To put our problem in an appropriate form for ADMM, rewrite \eqref{Eq_3A_2} as
\begin{eqnarray}
\label{eq:sdp2}
\mathop {\arg \min }\limits_{\substack{ \mathbf H \in \mathbb{C}^{M \times
N}, \\ \mathbf U_{2} \in \mathbb{C}^{(2M_2-1) \times (2M_1-1)}, \\
        \mathbf V_{2} \in \mathbb{C}^{(2N_2-1)\times (2N_1-1)} }} && \frac{1}{2}{\left\|  \bf{H} \bf{P} - {\bf{Y}} \right\|_F^{2}} + \frac{\gamma}{2M} {\text{Tr}}\left( {\mathbb T}_2(\mathbf U_{2}) \right) + \frac{\gamma}{2N} {\text{Tr}}\left( {\mathbb T}_2(\mathbf V_{2}) \right) + \mathbb{I}_\infty(\bf M \succeq 0), 
\end{eqnarray}
where $\mathbb{I}_\infty(z)$ is an indicator function that is 0 if $z$ is true, and infinity otherwise. Dualize the equality constraint via an augmented Lagrangian, we have
\begin{eqnarray}
\label{eq:AL}
{\cal L}_\rho(\mathbf U_{2}, \mathbf V_{2},\bf{H},\bm{\Upsilon},\bf M) &=& \frac{\gamma}{2M} {\text{Tr}}\left( {\mathbb T}_2(\mathbf U_{2}) \right) + \frac{\gamma}{2N} {\text{Tr}}\left( {\mathbb T}_2(\mathbf V_{2}) \right) + \frac{1}{2} \|\mathbf H \mathbf P - \mathbf Y \|_{F}^{2} + \mathbb{I}_\infty(\bf M \succeq 0) \nonumber \\
&& + \left\langle {{\bm \Upsilon},\bf M  - \left[ {\begin{array}{*{20}{c}}
                {{\mathbb T}_2(\mathbf U_{2})}& \bf H\\
                {\mathbf H}^H & {\mathbb T}_2(\mathbf V_{2})
                \end{array}} \right]} \right\rangle \nonumber \\ 
&& + \frac{\rho}{2} \left \| {\bf M  - \left[ {\begin{array}{*{20}{c}}
                {\mathbb T}_2(\mathbf U_{2}) & \mathbf H\\
                {\mathbf H }^H & {\mathbb T}_2(\mathbf V_{2})
                \end{array}} \right] }\right\|_F^2,
\end{eqnarray}
where $\bm \Upsilon$ is the dual variable, $\langle \mathbf \Upsilon, \mathbf M \rangle \triangleq \text{Re} \left( \text{Tr}(\mathbf M^H \mathbf \Upsilon) \right)$, $\rho > 0$ is the penalty parameter. The ADMM consists of the following update steps:
\begin{eqnarray}
\label{eq:UPDA}
(\mathbf U^{l+1}_{2}, \mathbf V^{l+1}_{2},\mathbf{H}^{l+1}) &=& \arg \min_{\substack{ \mathbf H \in \mathbb{C}^{M \times N}, \\ \mathbf U_{2} \in \mathbb{C}^{(2M_2-1) \times (2M_1-1)}, \\
                \mathbf V_{2} \in \mathbb{C}^{(2N_2-1)\times (2N_1-1)} } }{\cal L}_{\rho}(\mathbf U_{2}, \mathbf V_{2}, \mathbf H, {\mathbf \Upsilon}^{l}, {\mathbf M }^{l} ),\\
\label{eq:UPDA2}
\mathbf M^{l+1} &=& \arg \min_{\mathbf M \in \mathbb{C}^{(M+N)\times(M+N)} \succeq 0}{\cal L}_{\rho}(\mathbf U_2^{l+1},\mathbf V_2^{l+1},\mathbf H^{l+1},\mathbf \Upsilon^{l},\mathbf M),\\
\label{eq:UPDAUpsilon}
\mathbf \Upsilon^{l+1} &=& \mathbf \Upsilon^{l} + \rho \left( {\mathbf M^{l+1}  - \left[ {\begin{array}{*{20}{c}}
                {{\mathbb T}_2(\mathbf U_2^{l+1})}& \mathbf H^{l+1}\\
                (\mathbf H^{l+1})^H&{{\mathbb T}_2(\mathbf V_2^{l+1})}
                \end{array}} \right] } \right).
\end{eqnarray}

Now we derive the updates of \eqref{eq:UPDA} and \eqref{eq:UPDA2} in detail. For convenience, the following partitions are introduced:
\begin{eqnarray}
\mathbf M^l = \left[ {\begin{array}{*{20}{c}}
        {\mathbf M_{0}^l}&{\mathbf M_{2}^l}\\
        {(\mathbf M_{2}^l)^H}&{\mathbf M}_1^l
        \end{array}} \right],
\end{eqnarray}
\begin{eqnarray}
\label{eq:Upsilon}
\bm \Upsilon^l = \left[ {\begin{array}{*{20}{c}}
        {\bm \Upsilon_0^l}&{\bm \Upsilon_2^l}\\
        {(\bm \Upsilon_2^l)^H}&{\bm \Upsilon_1^l}
        \end{array}} \right],
\end{eqnarray}
where $\mathbf M_0^l$ and $\bm \Upsilon_0^l$ are $M \times M$ matrices, $\mathbf M_2^l$ and $\bm \Upsilon_{2}^l$ are $M \times N$ matrices, $\mathbf M_1^l$ and $\bm \Upsilon_1^l$ are $N \times N$ matrices. Computing the derivative of ${\cal L}_\rho(\mathbf U_2, \mathbf V_2,\mathbf{H},\mathbf{\Upsilon},\mathbf M)$ with respect to $\mathbf H$, $\mathbf U_2$ and $\mathbf V_2$, we have
\begin{eqnarray}
\nabla_{\mathbf{H}} {\cal L}_{\rho} &=& (\mathbf{H}\mathbf{P} - \mathbf{Y})\mathbf{P}^H - 2 \mathbf \Upsilon_2^l + 2\rho (\mathbf{H} - \mathbf M_{2}^l),\\
\nabla_{U_{2}(i,k)}{\cal L}_{\rho} &=& \left\{ \begin{array}{l}
\frac{{{\gamma}}}{2} + M_1\rho U_2(i,k) - {\text{Tr}}(\rho \mathbf M_0^l + \mathbf \Upsilon_0^l),\ i = k = 0,\\
(M_1 - i)(M_2 - k)\rho U_2(i,k) - \sum\limits_{m=0}^{M_2-i-1}{\text{Tr}}_{k}\left({\cal S}_{i,k}^{(1)}(\rho \mathbf M_{0}^l + \mathbf \Upsilon _{0}^l)\right),\ i \ne 0 \text{ or } k \ne 0,
\end{array} \right. \\
\nabla_{V_{2}(i,k)}{\cal L}_{\rho} &=& \left\{ \begin{array}{l}
\frac{{{\gamma}}}{2} + N_1\rho V_2(i,k) - {\text{Tr}}(\rho \mathbf M_1^l + \mathbf \Upsilon_1^l),\ i = k = 0,\\
(N_1 - i)(N_2 - k)\rho V_2(i,k) - \sum\limits_{m=0}^{N_2-i-1}{\text{Tr}}_{k}\left({\cal S}_{i,k}^{(2)}(\rho \mathbf M_{1}^l + \mathbf \Upsilon _{1}^l)\right),\ i \ne 0 \text{ or } k \ne 0,
\end{array} \right. 
\end{eqnarray}
where $U_2(i,k)$ and $V_2(i,k)$ are the $(i,k)$-th elements of $\mathbf U_2$ and $\mathbf V_2$, respectively. For $\mathbf X \in \mathbb{C}^{M \times M}$, ${\cal S}_{i,k}^{(1)}(\mathbf X)$ returns the $(i,k)$-th $M_1 \times M_1$ submatrix $\mathbf X_{i,k}$. For $\mathbf X \in \mathbb{C}^{N \times N}$, ${\cal S}_{i,k}^{(2)}(\mathbf X)$ returns the $(i,k)$-th $N_1 \times N_1$ submatrix $\mathbf X_{i,k}$. $\text{Tr}_{k}(\cdot)$ outputs the trace of the $k$-th sub-diagnal of the input matrix. ${\text {Tr}}_{0}(\cdot)$ outputs the trace of the input matrix.

By setting the derivatives to 0, $\mathbf{H}^{l+1}$, $\mathbf U_2^{l+1}$ and $\mathbf V_2^{l+1}$ can be updated by:
\begin{eqnarray}
\label{eq:UPDAx1}
\mathbf{H}^{l+1} &=& (\mathbf{Y}\mathbf{P}^{H} + 2\rho \mathbf M_{2}^l + 2 \mathbf \Upsilon_{2}^l) (\mathbf{P} \mathbf{P}^H + 2\rho \mathbf{I}_N)^{-1},\\
\label{eq:UPDAu}
\mathbf U_2^{l+1} &=& {\mathbb T}_2^*(\mathbf M_{0}^l + \mathbf \Upsilon_{0}^l/\rho) - \frac{\gamma}{2 M \rho} \mathbf e_1 , \\
\label{eq:UPDAv}
\mathbf V_2^{l+1} &=& {\mathbb T}_2^*(\mathbf M_{1}^l + \mathbf \Upsilon_{1}^l/\rho) - \frac{\gamma}{2 N \rho} \mathbf e_1 ,
\end{eqnarray}
where $\mathbf e_1 = [1,0,0,...,0]^T$, ${\mathbb T}_2^*(\cdot)$ denotes the adjoints of the map ${\mathbb T}_2(\cdot)$. Specifically, suppose $\mathbf Z = {\mathbb T}_2^*(\mathbf X )$  where $\mathbf Z = [\mathbf z_{-M_2+1},\mathbf z_{-M_2+2},...,\mathbf z_{M_2-1}]$ with $\mathbf z_i=[z_i(-M_1+1),z_i(-M_1+2),...,z_i(M_1-1)]^T$ when $\mathbf
X \in \mathbb{C}^{M\times M}$. Then we have
\begin{eqnarray}
\label{eq:adjoint}
z_i (k) = \frac{1}{(M_1-i)(M_2-k)}\sum_{m=0}^{M_1-i-1} \text{Tr}_{k} ({\cal S}_{i,m}^{(1)}(\mathbf X)),
\end{eqnarray}
for $i= -M_2+1,-M_2+2,...,M_2-1$ and $k= -M_1+1,-M_1+2,...,M_1-1$.  

The update of $\mathbf M$ is given by
\begin{eqnarray}
\label{eq:UPDAPhi}
\mathbf M^{l+1} = \arg\min_{\mathbf M  \in \mathbb{C}^{(M+N)\times(M+N)}
 \succeq 0} \left \| \mathbf M  - \mathbf{\tilde M}^{l+1} \right\|_F^2,
\end{eqnarray}
where
\begin{eqnarray}
\mathbf{\tilde M}^{l+1} = \left[ {\begin{array}{*{20}{c}}
        {{\mathbb T}_2(\mathbf U_2^{l+1})}& \mathbf{H}^{l+1}\\
        (\mathbf{H}^{l+1})^H&{{\mathbb T}_2(\mathbf V_2^{l+1})}
        \end{array}} \right] - \mathbf \Upsilon^{l+1}/\rho.
\end{eqnarray}
It is equivalent to projecting $\mathbf{\tilde M}^{l+1}$ onto the positive semidefinite cone. Specifically, the projection is accomplished by setting all negative eigenvalues of $\mathbf{\tilde M}^{l+1}$ to zero. Note that in ADMM the update of variables $\mathbf H$, $\mathbf U_2$, $\mathbf V_2$ and $\mathbf M$ are in closed-form. Compared to the off-the-shelf solvers such as SeDuMi \cite{sturm1999using} and SDPT3 \cite{toh1999sdpt3}, whose computational complexity is ${\mathcal O}\left( (M+N)^6 \right)$ in each iteration, the complexity of ADMM is ${\mathcal O}\left( (M+N)^3 \right)$ in each iteration, so it runs much faster. 

\section{The General Planar Array Case \label{Section_NUPA_chest}}
So far we have focused on the uniform planar array (UPA).  For mmWave beamformed FD-MIMO, because of the larger average inter-antenna element spacing, non-uniform planar array (NUPA) requires fewer elements than  UPA, whereby reducing the weight and cost of the system in large array applications. Also, the irregular spacing allows the antenna grid spacing to become larger than a half wavelength so it can effectively reduce the channel correlation and enhance multiplexing gain \cite{NUPA_mmWaveMIMOLinks}. Furthermore, there is a fundamental limitation of UPA, namely, the lower resolution of elevation AoA, which limits the UPA performance \cite{structured_non_unifromly_spaced_antenna_array}. 

In this section we consider the beamformed mmWave FD-MIMO channel estimation for NUPA. Define $\mathbf d_{t}= \frac{2}{\lambda} \left[ \left(d_{t,1}(1), d_{t,2}(1)\right) \ldots \left(d_{t,1}(N), d_{t,2}(N)\right) \right]$ as the normalized transmit antenna locations, where $\left( {d_{t,1}\left( i \right)}, {d_{t,2}\left( i \right)}\right)$ is the $i$-th transmit  antenna coordinate in a 2D planar surface. Similarly, $\mathbf d_{r}= \frac{2}{\lambda}\left[\left(d_{r,1}\left(1\right), d_{r,2}\left(1\right)\right) \ldots \left(d_{r,1}\left(M\right), d_{r,2}\left(M\right)\right)\right]$ is the normalized receive antenna locations where $\left( {d_{r,1}\left( i \right)}, {d_{r,2}\left( i \right)}\right)$ is the $i$-th receive  antenna coordinate in a 2D planar surface.  Then the steering responses of the transmit and receive arrays for the $l$-th path can be respectively written as \cite{nai2010beampattern}  
\begin{eqnarray} \label{Eq_non_uniform_sin}
\mathbf {a}_{\mathbf d_{t}}\left( \mathbf g_{l} \right) &=& \frac{1}{\sqrt{N}}
\left[ e^{j{2\pi} \left(\frac{ 2d_{t,1} (1) }{\lambda} g_{l,1} + \frac{2d_{t,2}( 1)}{\lambda}g_{l,2} \right) }  \cdots e^{j{2\pi} \left(\frac{2d_{t,1} (N)}{\lambda}g_{l,1} + \frac{2d_{t,2} (N)}{\lambda}g_{l,2}  \right)}   \right]^{T}, \\
\label{Eq_non_uniform_cos}
\mathbf {b}_{\mathbf d_{r}}\left( \mathbf f_{l} \right) &=& \frac{1}{\sqrt{M}}
\left[ e^{j{2\pi} \left(\frac{ 2d_{r,1} (1) }{\lambda} f_{l,1} + \frac{2d_{r,2}(
                1)}{\lambda}f_{l,2} \right) }  \cdots e^{j{2\pi} \left(\frac{2d_{r,1}(M)}{\lambda}f_{l,1} + \frac{2d_{r,2} (M)}{\lambda}f_{l,2}  \right)}   \right]^{T} .
\end{eqnarray}
With \eqref{Eq_non_uniform_sin} and \eqref{Eq_non_uniform_cos}, the channel matrix $\mathbf H  $ of NUPA is given by \eqref{Eq_2_2} with array responses $\mathbf a\left( \mathbf g_{l} \right)$  and $\mathbf b \left( \mathbf f_{l} \right)$  replaced by $\mathbf {a}_{\mathbf d_{t}}\left( \mathbf g_{l} \right)$ and $\mathbf {b}_{\mathbf d_{r}}\left( \mathbf f_{l} \right)$, respectively. %The data $\mathbf y $ is then given by \eqref{vector_data_y} with $\mathbfH$ %replaced by its new definition.

The atom for NUPA is then defined as 
\begin{eqnarray} \label{atom4D_definition_var}
\mathbf{q}_{\rm NU}\left( \mathbf{ g}, \mathbf{ f} \right) = \mathbf{a}_{\mathbf d_{t}}^{*}\left(\mathbf{
        g}\right) \otimes \mathbf{ b}_{\mathbf d_{r}}\left(\mathbf{ f}\right).
\end{eqnarray}
And  the atom set for NUPA is given by
\begin{eqnarray}
\label{eq:atom_dic}
\mathcal{A}_{\text{NU}} \triangleq \left\{ \mathbf{q}_{\rm NU} \left( \mathbf{ g}, \mathbf{ f}  \right), \mathbf{ g} \in [\frac{-1}{2},\ \frac{1}{2}) \times [-\frac{1}{2},\ \frac{1}{2}),\ \mathbf{ f} \in [\frac{-1}{2},\ \frac{1}{2}) \times [\frac{-1}{2},\ \frac{1}{2})  \right\}.
\end{eqnarray}
The  atomic norm $ \left\| \mathbf h \right\|_{\cal A_{\text{NU}}}$  for any $\mathbf h = \text{vec}\left(\mathbf H \right)  $  is then given by 
\begin{eqnarray}
\label{eq:atomnorm_nupa_2D}
\left\| \mathbf h \right\|_{\cal A_{\rm NU}} = \inf_{\substack{
                \mathbf f_{l} \in [-\frac{1}{2},\ \frac{1}{2}) \times [-\frac{1}{2},\ \frac{1}{2}),\\ \mathbf g_{l} \in [-\frac{1}{2},\ \frac{1}{2}) \times [-\frac{1}{2},\ \frac{1}{2}),\\ \sigma_{l} \in \mathbb{C}
}} \left\{ \left. \sum\limits_{l} {\left\vert\sigma_{l}\right\vert}
\right| \mathbf h = \sum\limits_{l} \sigma_{l} \mathbf q_{\rm NU} \left( \mathbf g_{l}, \mathbf f_{l} \right)  \right\}.  
\end{eqnarray}

To estimate the channel, we propose to solve the following optimization problem 
\begin{eqnarray}\label{object_function}
\min_{\mathbf h} \mu\left\Vert \mathbf h \right\Vert_{\mathcal{A}_{\text{NU}}}
+ \frac{1}{2} \left\Vert \mathbf{y} - \sqrt{P_t} \left(\mathbf P^{T} \otimes \mathbf I_{M}\right)
\mathbf h \right\Vert_{2}^2.
\end{eqnarray}
Note that the atom defined in \eqref{atom4D_definition_var} is not based on uniform sampling, and consequently the atomic norm in \eqref{eq:atomnorm_nupa_2D} does not have the equivalent SDP form as in \eqref{Eq_2D_atomic_norm_sdp} or \eqref{Eq_3A_2}. Hence,  \eqref{object_function} cannot be solved via convex optimization. According to Corollary 2.1 of \cite{li2016approximate}, \eqref{object_function} shares the same optimum as the following optimization problem
\begin{eqnarray}
\label{eq:p2}
\min_{\substack{
		\mathbf f_{l} \in [-\frac{1}{2},\ \frac{1}{2}) \times [-\frac{1}{2},\ \frac{1}{2}),\\ \mathbf g_{l} \in [-\frac{1}{2},\ \frac{1}{2}) \times [-\frac{1}{2},\ \frac{1}{2}),\\ \sigma_{l} \in \mathbb{C}
	}} \Gamma\left( \{ \mathbf g_{l}, \mathbf f_{l},  \sigma_{l} \} \right) = \mu\left\Vert \bm{\sigma} \right\Vert_{1}
+ \frac{1}{2} \left\Vert \mathbf{y} - \sqrt{P_t} \left(\mathbf P^{T} \otimes \mathbf I_{M}\right) \sum_{l=1}^{L}\mathbf q_{\rm NU}\left( \mathbf g_{l}, \mathbf f_{l} \right){\sigma_{l}} \right\Vert^{2}_{2}.
\end{eqnarray}
Since the problem given by \eqref{eq:p2} is nonconvex, we will employ  a gradient-descent algorithm to obtain its local optimum. In practice, $L$ is unknown, so we initialize $\mathbf q \left(\mathbf{g}_{l}, \mathbf{f}_{l}\right)$ on $\tilde L^{0}$ grid points such that $L \leq \tilde L^{0} \leq MP$, where $P$ is the number of training beams defined in \eqref{Eq_2_4}. For example, let each $\mathbf g_{l}$ and $\mathbf f_{l}$ be taken from a uniform grid of $N_{G}$ points with $\tilde L^{0}= N_{G}^4 \leq MP$, i.e., $g_{l,i}^{0}$ and $f_{l,i}^{0}$ are uniformly taken from $[-1/2,1/2)$ for $i = 1,2$ and $1 \leq l \leq N_G^4$, where the supercript $^0$ indicates iteration $0$, i.e.,  initialization. Let $\mathbf \Omega^{0} = \left\{ (\mathbf g_{l}^{0}, \mathbf f_{l}^{0})_{1 \leq l \leq {\tilde L}} \right\}$. The   initial value of $\bm \sigma^{0} = \left[ \sigma_{1}^{0} \ldots \sigma_{\tilde L}^{0}\right]^{T}$ can then be obtained by the least-squares (LS) estimate 
\begin{eqnarray}\label{Eq_LS}
\bm {\sigma}^{0} = \left( \left(\mathbf
P^{T} \otimes \mathbf I_{M}\right) \left[ \mathbf{q}_{\rm NU}\left( \mathbf{
        g}^{0}_{1}, \mathbf{f}^{0}_{1} \right) \ldots \mathbf{q}_{\rm NU}\left( \mathbf{
        g}^{0}_{\tilde L}, \mathbf{ f}^{0}_{\tilde L} \right) \right] \right)^{\dagger}\mathbf
y,
\end{eqnarray}
where $^{\dagger}$ indicates the pseudo inverse of the matrix. Then the gradient descent method is used to find the local optimum. We use superscript $k$ to denote the quantities in the  $k$-th iteration. Then the gradient descent search proceeds as follows 
\begin{eqnarray} \label{subgradient_g}
g_{l,i}^{k+1} &=&  \left[g_{l,i}^{k} - \kappa^k \nabla_{g_{l,i}} \Gamma\left( \{\mathbf g_{l}^k, \mathbf f_{l}^k, \sigma_{l}^k \} \right) \right]^{\frac{1}{2}}_{-\frac{1}{2}}, \\ \label{subgradient_f}
f_{l,i}^{k+1} &=&  \left[f_{l,i}^{k} - \kappa^k \nabla_{f_{l,i}} \Gamma\left( \{ \mathbf g_{l}^k, \mathbf f_{l}^k, \sigma_{l}^k \} \right) \right]^{\frac{1}{2}}_{-\frac{1}{2}}, \\ \label{subgradient_sigma}
\sigma_{l}^{k+1} &=&  \sigma_{l}^{k} - \kappa^k \nabla_{\sigma_{l}} \Gamma\left( \{ \mathbf g_{l}^k, \mathbf f_{l}^k, \sigma_{l}^k \} \right), 
\end{eqnarray}
for $l=1,\ldots, \tilde L^{k}$ and $i=1, 2$, where $\kappa^k$ is the step size that can be obtained via Armijo line search \cite{boumal2015low} and $\left[ x \right]^{a}_{b}$ defines the operator that outputs $x = \mod\left(x, a \right) $ when  $x < b$, and outputs $x = \mod\left(x, b \right) $ when  $x > a$, $\mod(a, b)$ defines the modulo operator. Specifically, in the $k$-th iteration, $\kappa^k$ is initialized as $\kappa^k = \bar \kappa$. If $\Gamma\left( \{ \mathbf g_{l}^{k+1}, \mathbf f_{l}^{k+1},  \sigma_{l}^{k+1} \} \right) \geq \Gamma\left(  \{ \mathbf g_{l}^{k}, \mathbf f_{l}^{k},  \sigma_{l}^{k} \} \right)$, then $\kappa^k$ is updated by multiplication with a constant $0 < \alpha < 1$, i.e., $\kappa^k \leftarrow \alpha \kappa^k$. The gradients are calculated
\begin{eqnarray}
\nabla_{g_{l,i}} \Gamma\left( \{ \mathbf g_{l}, \mathbf f_{l}, \sigma_{l} \} \right) &=& \mathcal{R}\left\{ \sigma_{l} \left(\mathbf{\bar P}\sum_{l=1}^{\tilde L} \mathbf {q}_{\rm NU}\left( \mathbf{g}_{l},
\mathbf {f}_{l}\right)\sigma_{l} - \mathbf{y}\right)^{H} \mathbf{\bar P} \left(\mathbf a^{*}_{\mathbf d_{t,i}}\left(\mathbf{g}_{l}\right) \otimes \mathbf{b}_{\mathbf d_{r}}\left(\mathbf{f}_{l}\right)\right)   \right\}, \label{eq:subgradient_g} \\
\nabla_{f_{l,i}} \Gamma\left( \{ \mathbf g_{l}, \mathbf f_{l}, \sigma_{l} \} \right) &=& \mathcal{R}\left\{\sigma_{l} \left(\mathbf{\bar P} \sum_{l=1}^{\tilde L} \mathbf {q}_{\rm NU}\left( \mathbf{g}_{l},
\mathbf {f}_{l}\right)\sigma_{l} - \mathbf{y}\right)^{H} \mathbf{\bar P}\left( \mathbf a^{*}_{\mathbf d_{t}}\left(\mathbf{g}_{l}\right)
\otimes \mathbf{b}_{\mathbf d_{r,i}}\left(\mathbf{f}_{l}\right) \right) \right\}, \label{eq:subgradient_f} \\
\nabla_{\sigma_{l}} \Gamma\left( \{ \mathbf g_{l}, \mathbf f_{l}, \sigma_{l} \} \right) &=& \mu\frac{\sigma_l}{2|\sigma_l| } + \frac{1}{2} \left(\mathbf{\bar P} \sum_{l=1}^{\tilde L} \mathbf {q}_{\rm NU}\left( \mathbf{g}_{l},
\mathbf {f}_{l}\right)\sigma_{l} - \mathbf{y}\right)^T\left( \mathbf{\bar P}\mathbf q_{\rm NU}\left( \mathbf g_{l}, \mathbf f_{l} \right)  \right)^{*}, \label{eq:subgradient_sigma}
\end{eqnarray}
where $\mathcal R \left\{ \cdot \right\}$ returns the real part of the input, 
\begin{eqnarray}
\mathbf {\bar P} &=& \sqrt{P_{t}}\left(\mathbf P^{T} \otimes \mathbf I_{M}\right), \\
\mathbf a_{\mathbf d_{t,i}}\left(\mathbf{g}_{l}\right) &=& \left( \frac{j{2\pi}}{\lambda} \left[ d_{t,i}(1), \ldots, d_{t,i}(N) \right]^{T}\right) \circ \mathbf{a}_{\mathbf d_{t}}\left(\mathbf{g}_{l}\right), \\
\mathbf b_{\mathbf d_{r,i}}\left(\mathbf{f}_{l}\right) &=& \left( \frac{j{2\pi}}{\lambda}
\left[ d_{r,i}(1), \ldots, d_{r,i}(M) \right]^{T}\right) \circ \mathbf{b}_{\mathbf
	d_{r}}\left(\mathbf{f}_{l}\right),
\end{eqnarray}
and $\circ$ denotes Hadamard product.  The derivations of \eqref{eq:subgradient_g} - \eqref{eq:subgradient_sigma} are given in the Appendix. 
To accelerate the convergence, we introduce a pruning step to remove the atoms whose coefficients are smaller than a threshold during each iteration. Specifically, at the $k$-th iteration, if $|\sigma_l^k|<\eta^{k}$ where $\eta^{k}$ is a given threshold at the $k$-th iteration, then $l$-th path are removed from the set and number of estimated paths is decreased by one, i.e., $\mathbf \Omega^{k} \leftarrow \mathbf \Omega^{k} \setminus \left\{\left(\mathbf{g}^{k}_{l}, \mathbf{f}^{k}_{l} \right)\right\}$ and $\tilde L^{k} \leftarrow \tilde L^{k} -1$ at the $k$-th iteration. The algorithm stops when $ \left\| \mathbf{h}^{k+1} - \mathbf{h}^{k}\right\|  < \varepsilon$, where $\mathbf{h}^{k}=\sum_{l=1}^{\tilde L^{k}} \mathbf q_{\rm NU}\left( \mathbf{g}^{k}_{l}, \mathbf{f}^{k}_{l} \right){\sigma^{k}_{l}} $ denotes the channel estimation at the $k$-th iteration.

\section{Simulation Results\label{Section_Simulation}}
\subsection{Simulation Setup} 
In this section, we  evaluate the performance of the proposed channel estimators  for mmWave FD-MIMO links with UPA or NUPA. We compare the channel estimation performance of the proposed algorithm with some existing algorithms including the 4D-MUSIC \cite{2D_MUSIC_mmWave} and the orthogonal matching pursuit (OMP) \cite{OMP}.  The simulation parameters are set as follows.
\begin{enumerate} [1,]
        \item The numbers of transmit and receive antenna are $N=16$ and $M=16$, respectively. For UPA,  we set $N_{1}=4$, $N_{2}=4$, $M_{1}=4$ and $M_{2}=4$.
        \item In the UPA case, the DFT codebooks at the transmitter for elevation and azimuth  are given by 
        \begin{eqnarray}
        \mathbf P_{1} = \left[ \mathbf c_{N_{1}}\left( \psi_{1,0} \right)\ \mathbf c_{N_{1}}\left( \psi_{1,1} \right) \cdots \mathbf c_{N_{1}}\left( \psi_{1,P_{1}-1} \right) \right] \in \mathbb{C}^{N_{1} \times P_{1} },  \nonumber \\
        \mathbf P_{2} = \left[ \mathbf c_{N_{2}}\left( \psi_{2,0} \right)\  \mathbf c_{N_{2}}\left( \psi_{2,1} \right) \cdots \mathbf c_{N_{2}}\left( \psi_{2, P_{2}-1} \right) \right] \in \mathbb{C}^{N_{2} \times P_{2}}, \nonumber
        \end{eqnarray}
        where $P_{1}$ and $P_{2}$ are the sizes of elevation and azimuth codebooks, respectively. The DFT angles are $\psi_{1,i}=  \frac{i }{P_{1}}$ for $i=0, \ldots, P_{1}-1$ and $\psi_{2,i}= \frac{i}{P_{2}}$ for $i=0, \ldots, P_{2}-1$. We take the Kronecker product of  $\mathbf P_{1}$ and $\mathbf P_{2}$ to form the product codebook $\mathbf P= \mathbf P_{1} \otimes \mathbf P_{2}$ with size $P=P_{1}P_{2}$. Each  beamforming vector has a unit norm, i.e., $\left\Vert \mathbf p_{p}\right\Vert = 1$ for $p=1,\ldots,P$ and $\text{rank} \left(\mathbf P\right) = P$.  
        
        \item The weight factor in \eqref{Eq_2D_atomic_norm_sdp} and \eqref{Eq_3A_2} is set as $\mu=\sigma_{w}\sqrt{MN\log\left( MN \right)}$. The weight for the augmented Lagrangian  in \eqref{eq:AL} is set as $\rho=0.05$.    
        
        \item  $\mathbf g_{l}$ and $\mathbf f_{l}$ for each path are assumed to uniformly take values in $[-\frac{1}{2}, \frac{1}{2}) \times [-\frac{1}{2}, \frac{1}{2})$. The number of paths $L=3$. 
        \item The signal power is controlled by the signal-to-noise ratio (SNR) which is defined as $\text{SNR}=\frac{P_{t}}{\sigma^{2}_{w}} $ with $\sigma^{2}_{w} = 1$. 
        \item \label{sim:itemNUPA} For NUPA, we use circular arrays for both transmitter and receiver with $N$ and $M$ antenna elements located on the 2D plane,  respectively. Specifically,  the $n$-th transmit antenna location is set as $d_{t,1}(n)=R_{t}\cos\left( \chi_{n} \right), d_{t,2}(n)=R_{t}\sin\left( \chi_{n} \right),\ n=1,2, \ldots, N$, where $\chi_{n}=2\pi\left({n\over N}\right)$ is the angular position of the $n$-th element and  $R_{t}$ is the radius of the transmit array. Similarly, the $m$-th receive antenna location is $d_{r,1}(m)=R_{r}\cos\left( \chi_{m} \right), d_{r,2}(m)=R_{r}\sin\left( \chi_{m} \right),\ m=1,2, \ldots, M$, where $\chi_{m}=2\pi\left({m\over M}\right)$ is the angular position of the $m$-th element and $R_{r}$ is the radius of the receive array.     
        \item For the gradient descent algorithm, we set $\tilde L^0 = MP$ as the initial value in both UPA and NUPA cases. The pruning threshold in the $k$-th step is set as  $\eta^{k}= 0.7 \max_{1 \leq l \leq \tilde L^k}\left\{ \sigma^{k}_{l}\right\}  $.
        \item For the OMP and 4D-MUSIC algorithms, the AoD and AoA grid points are set as $\bar \vartheta_{i} = \frac{\left( i-1 \right)2\pi}{N_{G}}-\pi$, $\bar \varphi_{i}= \frac{\left( i-1 \right)2\pi}{N_{G}}-\pi$ and $\bar \theta_{i} = \frac{\left( i-1 \right)2\pi}{N_{G}}-\pi$, $\bar \phi_{i} =\frac{\left( i-1 \right)2\pi}{N_{G}}-\pi$, respectively, for $i=1, \ldots N_{G}$. 
        \item In the simulation, we use the CVX package \cite{cvx} to compute the 4D atomic norm-based estimator.
\end{enumerate}

\subsection{Performance Evaluation}   
% i.e., and average user rate $R=\mathbb{E}\left\{ \log\det\left[ \mathbf I_{M} + \frac{\text{SNR}}{N} \mathbf{\hat H} \mathbf{\hat H}^{H}\right] \right\}
% $ (bps/Hz).  
% 
We use the normalized mean square error (NMSE), i.e., $\text{NMSE} = \mathbb{E} \left\{\frac{\left\Vert \mathbf{\hat H} - \mathbf H\right\Vert^{2}_{F}}{\left\Vert\mathbf H\right\Vert^{2}_{F}}\right\}$ as the channel estimation performance metric. The NMSE statistics across different SNRs with different test setups are evaluated. Each curve is obtained by averaging over 100 realizations. First we compare the channel estimation performance of different algorithms under the UPA setting. Then we show the channel estimation performance for NUPA with the proposed gradient descent estimator and compare it with the 4D-MUSIC and OMP algorithms. 

The computational complexity of the proposed   approximate 4D atomic-norm-based channel estimator  is $\mathcal{O} (\left(M+N \right)^3)$ per-iteration. The computational complexity of the MUSIC estimator is $\mathcal{O}\left( \left(NM\right)^{3} + N^{4}_{G}\left(NM\right)^2\right)$ where $\mathcal{O}\left(\left(NM\right)^{3}\right)$ is for eigen decomposition and $\mathcal{O}\left(N^{4}_{G}\left(NM\right)^2\right)$ is for grid search. The complexity of the OMP estimator is $\mathcal{O}\left(N^{4}_{G}\left(NM\right)^2\right)$ per iteration. The complexity of proposed gradient descent estimator is $\mathcal{O}\left( M \left( N+P\right) \right) $ per iteration. 

\subsubsection{Convergence Behavior of the Proposed Channel Estimators}
We illustrate the convergence of the proposed   ADMM implementation of the approximate 4D atomic-norm-based
channel estimator through a simulation example. We compare the NMSE of the  ADMM channel estimator with that of the CVX solver \cite{cvx} that directly solves \eqref{Eq_3A_2}.  As can be seen from Fig. \ref{figADMM}, the proposed ADMM channel estimator converges to the solution given by the CVX after $300$-$400$ iterations for different SNR. It is worth noting that the ADMM runs much faster than the CVX solver because the calculation in each iteration is in closed-form. 
\begin{figure}[ht]
        \centering
        \subfigure[SNR$=4$dB]{\includegraphics[width=0.4\linewidth]{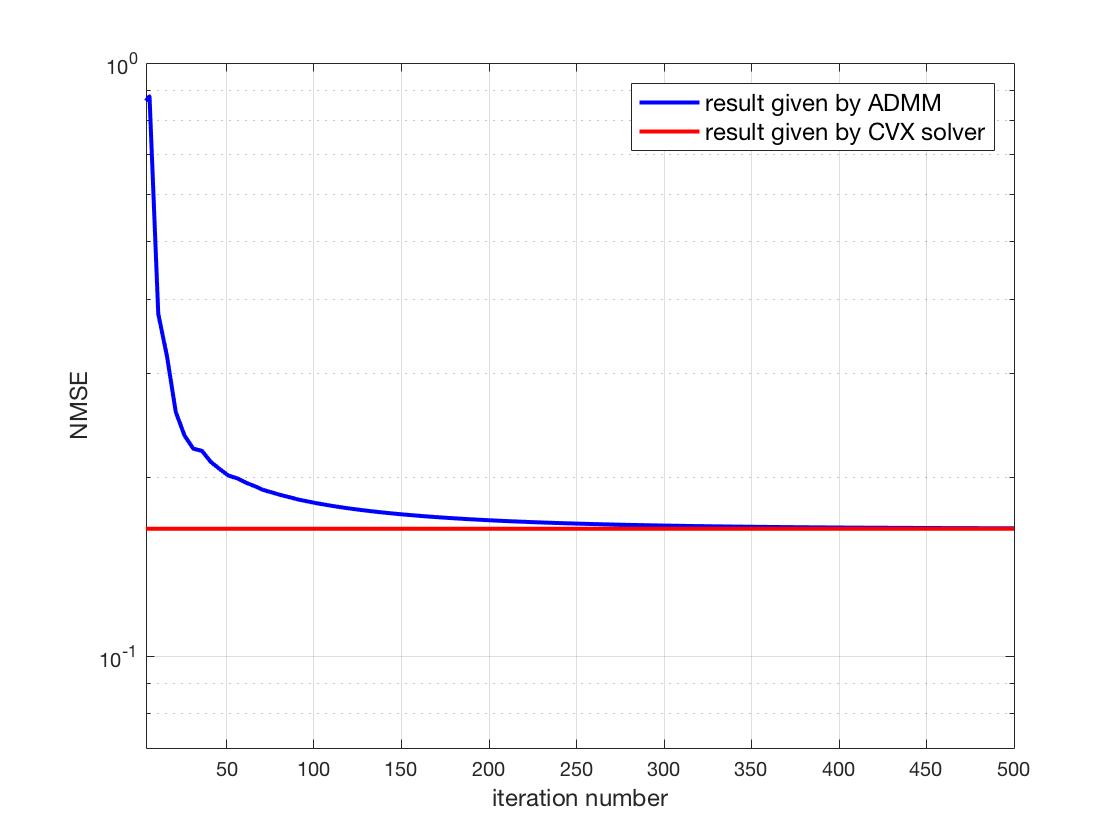}}\label{figADMM:subfiga}
          \subfigure[SNR$=10$dB]{\includegraphics[width=0.4\linewidth]{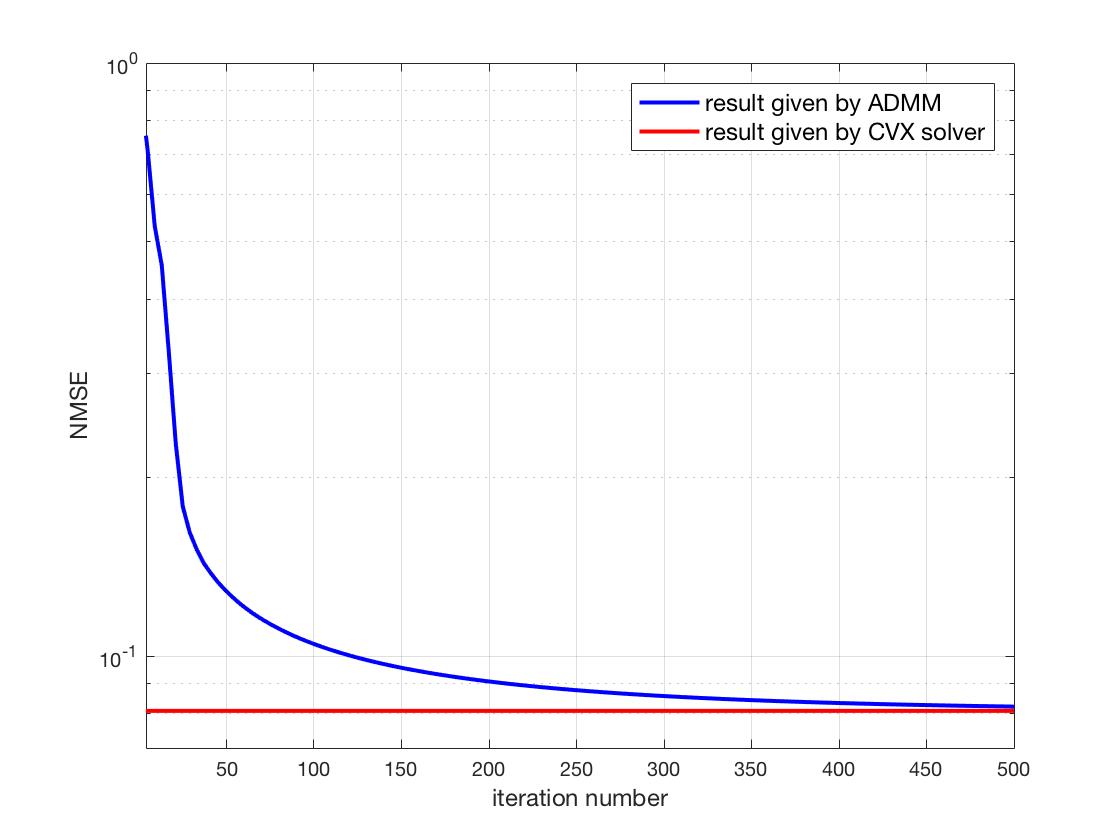}}\label{figADMM:subfiga}
        \caption{Convergence of proposed ADMM channel estimator with different SNR.} \label{figADMM}
\end{figure} 
We then show the convergence behavior and the number of estimated paths of the proposed gradient-descent-based channel estimator for UPA and NUPA in Fig. 3. It is seen that the algorithm converges within $1500$-$2000$ iterations for different SNR. The estimated number of paths is more accurate at higher SNR  when the algorithm converges, as more spurious frequencies arise when the noise is stronger. It is also worth noting that the computational complexity of the gradient descent method is lower than that of the ADMM, but the overall running time is higher because it takes more iterations. 
\begin{figure}[ht]
        \centering
        % -------------- SUBFIGURE -1 -------------- %
        \subfigure[]{\includegraphics[width=0.35\linewidth]{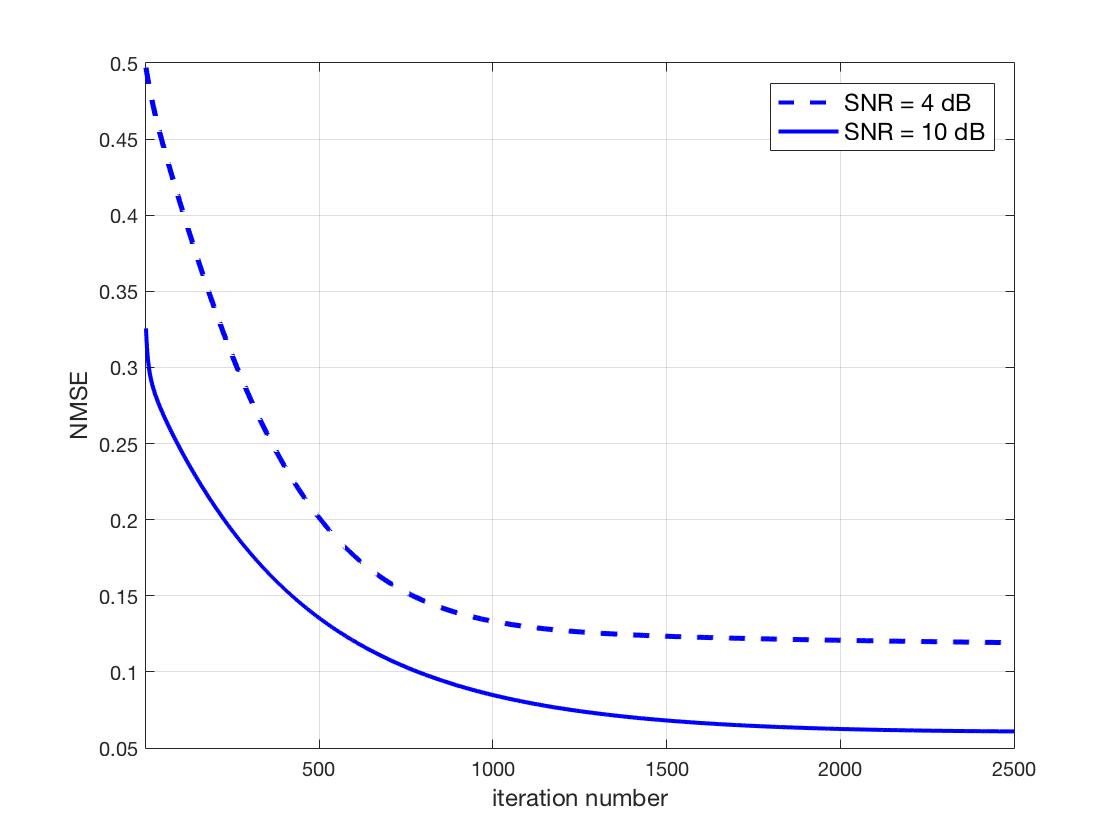} \label{figUPA_UCA:subfiga}}
        % -------------- SUBFIGURE -2 -------------- %
        \subfigure[]{\includegraphics[width=0.35\linewidth]{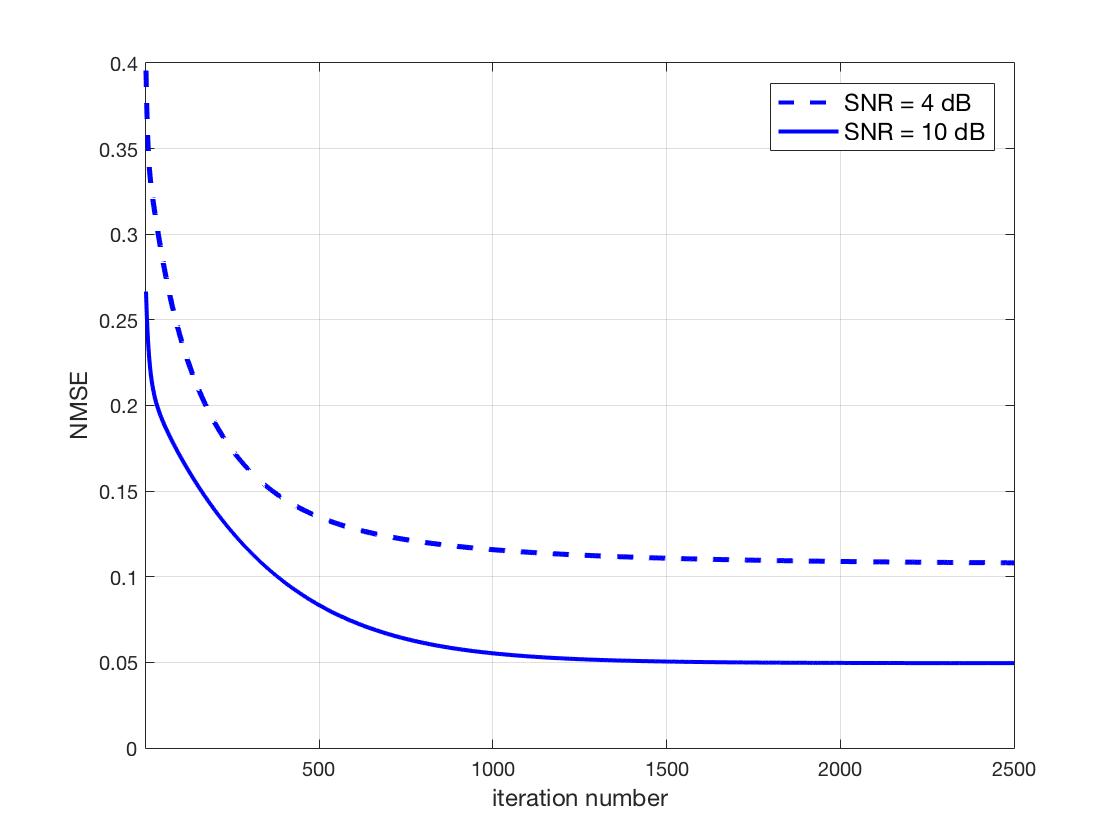}
                \label{figUPA_UCA:subfigb}}
         % -------------- SUBFIGURE -3 -------------- %
         \subfigure[]{\includegraphics[width=0.35\linewidth]{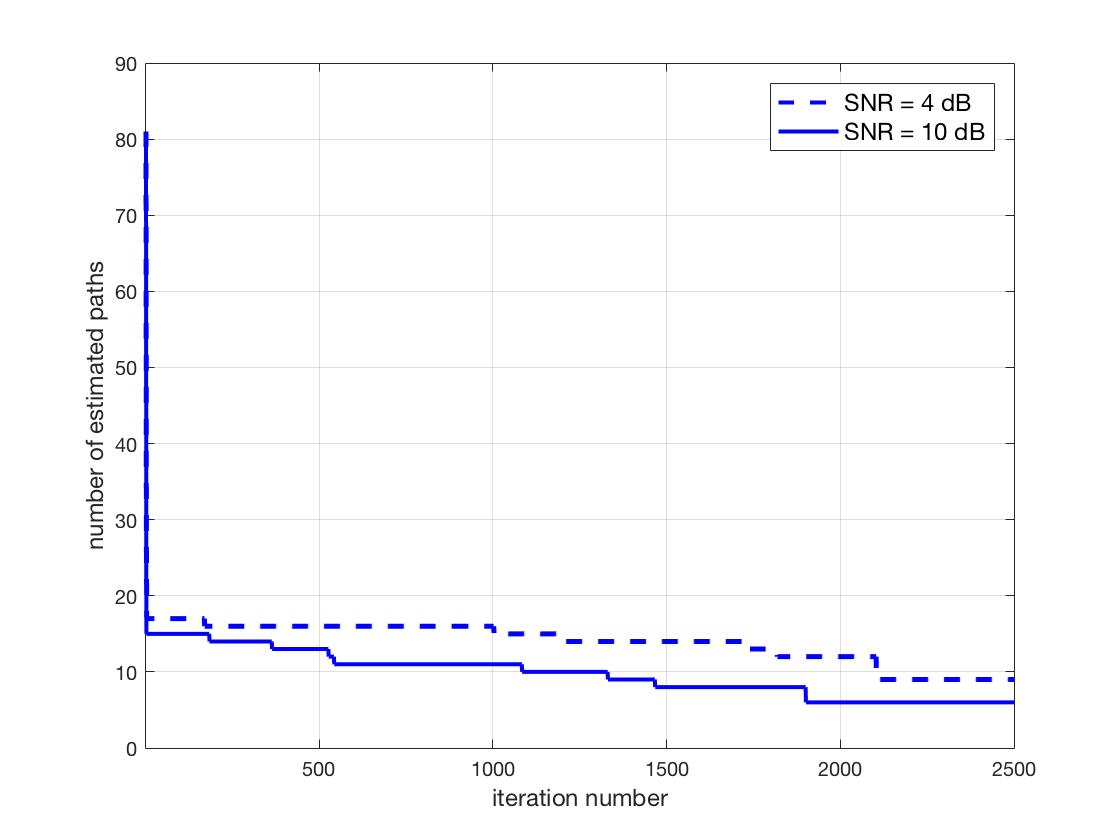} \label{figGD_Purning:subfiga}}
         % -------------- SUBFIGURE -4 -------------- %
         \subfigure[]{\includegraphics[width=0.35\linewidth]{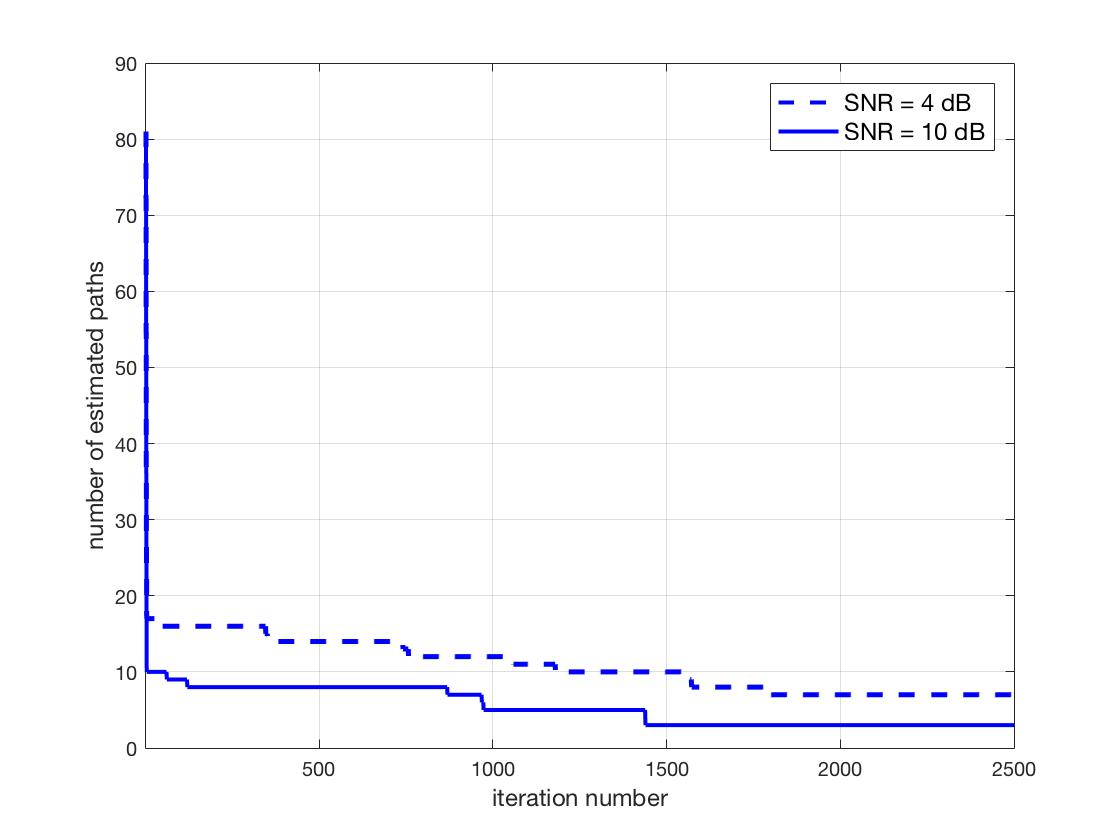}
                \label{figGD_Purning:subfigb}}   
        \caption[UPA_UCA]{Convergence and the number of estimated paths of the proposed gradient descent algorithm for (a)(c) UPA and (b)(d) NUPA.}
        \label{figUPC_UCA_convergence}
\end{figure}
\subsubsection{Comparison of On-grid and Off-grid Algorithms}
We compare the proposed off-grid channel estimator with two existing on-grid  approaches including OMP and MUSIC. For the on-grid algorithms, the continuous  AoA and AoD parameter spaces are discretized into a finite set of grids covering $[-\pi, \pi]$, and the estimation performance  improves with higher grid  resolution (i.e., larger $N_{G}$). However, higher grid  resolution leads to higher  computational complexity. 

In Fig. \ref{figGrid}, the NMSE and running time of different channel estimators are plotted against $N_{G}$. In this simulation, we use CVX solver to compute the 4D atomic-norm-based estimator and the ADMM algorithm to compute the approximate 4D atomic-norm-based  estimator. It is worth noting that the proposed approximate 4D atomic-norm-based estimator has the smallest complexity while its NMSE is much smaller than those of the on-grid algorithms. As the algorithm does not require the grids, its computational complexity does not change with $N_G$. In addition, its NMSE performance is only slightly worse than the 4D atomic-norm-based channel estimator, indicating that the performance loss caused by the approximation of  $ \left\Vert \mathbf H\right\Vert_{\mathcal{A}_M}$ by $\text{SDP}(\mathbf H)$ is small.   
\begin{figure}[H]
        \centering
        % -------------- SUBFIGURE -1 -------------- %
        \subfigure[]{\includegraphics[width=0.45\linewidth]{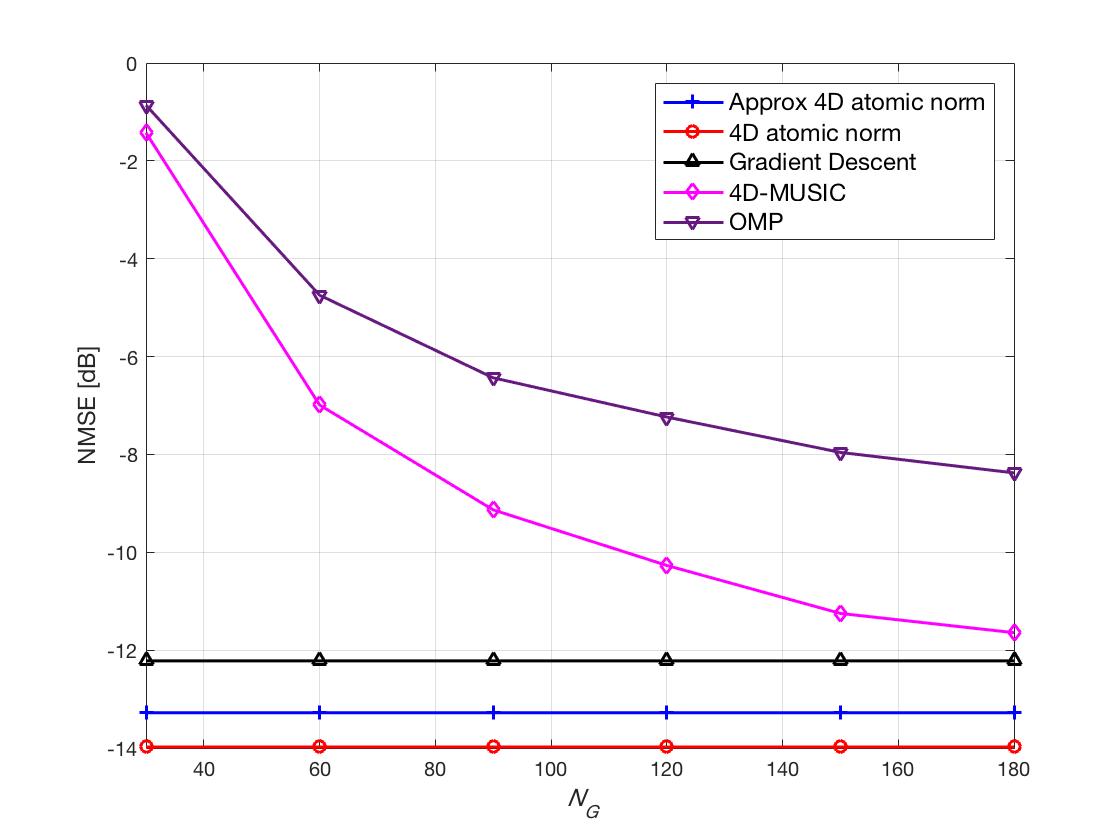} \label{figGrid:subfiga}}
        % -------------- SUBFIGURE -2 -------------- %
        \subfigure[]{\includegraphics[width=0.45\linewidth]{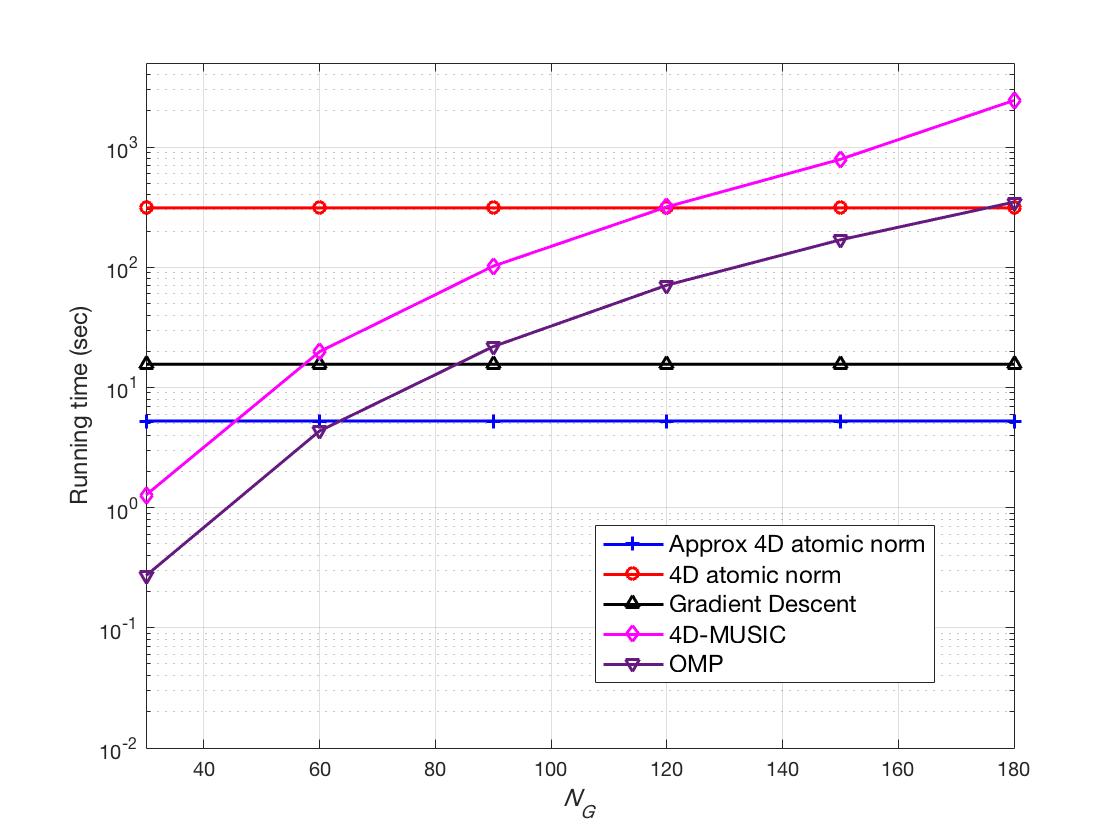}
                \label{figGrid:subfigb}}
        \caption[UPA_UCA]{Comparison of channel  estimation performance and running time against grid size, SNR$=10$ dB. (a) NMSE performance; (b) running time.}
        \label{figGrid}
\end{figure}
\subsubsection{Channel Estimation Performance}
Fig. \ref{figUPA_NMSE}  plots the NMSE curves as a function of SNR for different channel estimators under UPA. The number of grid points are set as $N_{G}=90, 180$ for 4D-MUSIC-based and OMP-based channel estimators.

It is seen that 4D atomic-norm-based and approximate 4D atomic-norm-based estimators outperform the 4D-MUSIC-based and OMP-based estimators. Meanwhile, the 4D atomic-norm-based channel estimator achieves better performance than the approximate 4D atomic-norm-based channel estimator by about from $0.5$ - $0.8$ dB. And the approximate 4D atomic-norm-based channel estimator outperforms the gradient-descent-based algorithm by more than $1.0$ dB. 
\begin{figure}[H]
        \centering
        \includegraphics[width=0.5\linewidth]{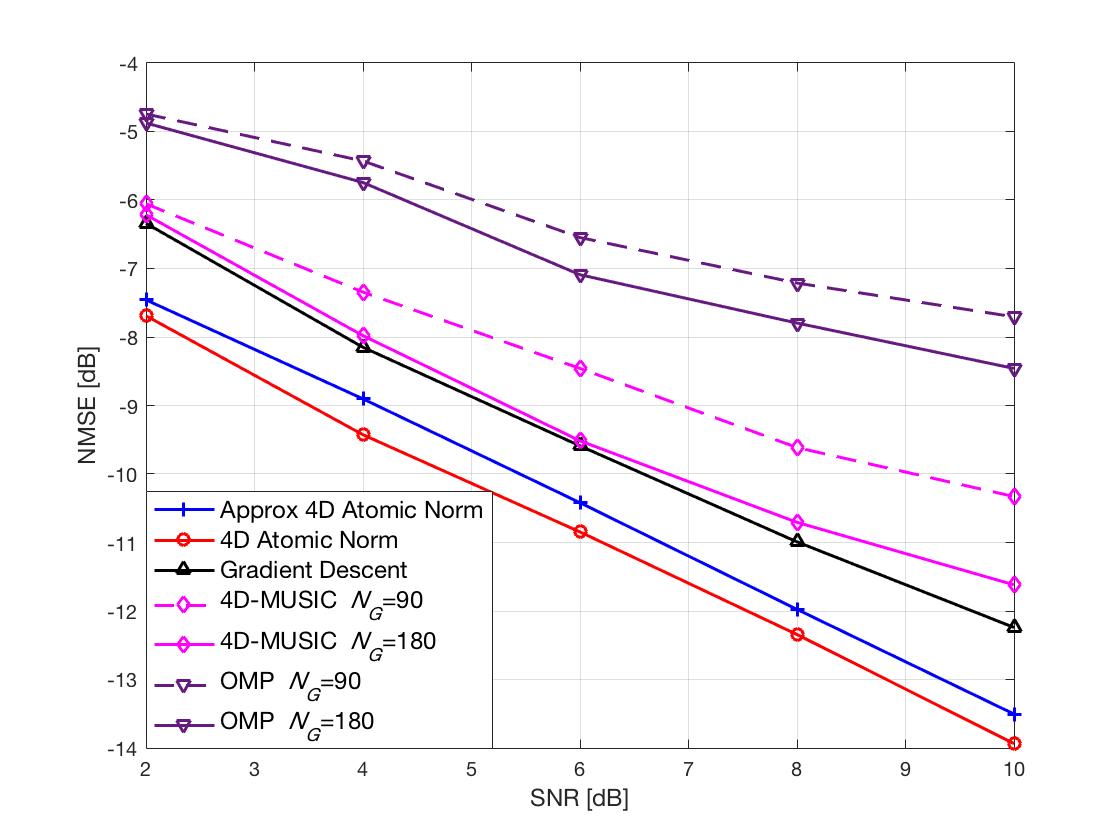}
        \caption[UPA NMSE]{The NMSE performance as a function of SNR for UPA.}
        \label{figUPA_NMSE}
\end{figure}
In Fig. \ref{figNUPA_NMSE}, we plot the NMSE curves as a function of SNR for different channel estimators under NUPA.  It is seen that  the  proposed gradient-descent-based channel estimator outperforms the 4D-MUSIC and  OMP-based channel estimators across the range of SNRs from $2$ to $10$ dB. This is because the proposed gradient-descent-based channel estimator optimizes the frequency basis in each iteration, so it outperforms the on-grid algorithms. 
\begin{figure}[H]
        \centering
        \includegraphics[width=0.5\linewidth]{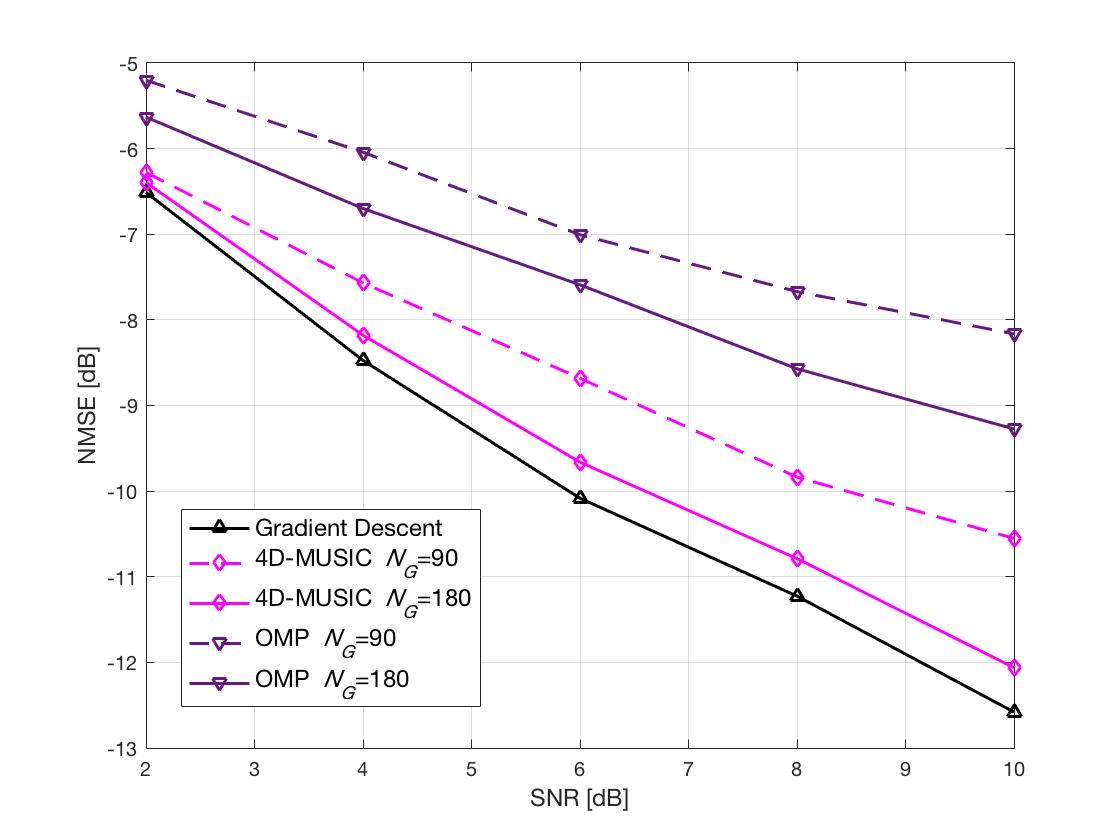} 
        \caption[UPA NMSE]{The NMSE performance as a function of SNR for NUPA.}
        \label{figNUPA_NMSE}
\end{figure}

\section{Conclusions\label{Section_Conclusions}}
In this paper, we have proposed  new channel estimation schemes for mmWave beamformed FD-MIMO systems based on atomic norm minimization  under both UPA and NUPA settings. For the UPA case, we approximate the original large-scale 4D atomic norm minimization problem using a semi-definite program (SDP) containing two decoupled two-level Toeplitz matrices which is then solved by an ADMM-based fast algorithm. For the NUPA case, a gradient descent-based algorithm is provided to obtain a suboptimal solution. Simulation results show that the proposed atomic norm based mmWave FD-MIMO channel estimators provide better performance compared to the existing methods based on compressed sensing and MUSIC algorithms.

\appendix
\subsection{Derivation for \eqref{eq:subgradient_g} and \eqref{eq:subgradient_f} }
For clarity, define $\mathbf{\bar P} = \sqrt{P_{t}}\left(\mathbf P^{T} \otimes \mathbf I_{M}\right)$. Then the gradient with respect to $g_{l,i}$ can be calculated by
\begin{eqnarray} \label{eq:derive_subgradient_g}
\nabla_{ g_{l,i}}  \Gamma\left( \{ \mathbf g_{l}, \mathbf f_{l}, \sigma_{l} \} \right) = \frac{1}{2} \frac{\partial
        \left(\mathbf{y} - \mathbf {\bar P} \mathbf{h} \right)^H\left(\mathbf{y}
        - \mathbf {\bar P} \mathbf{h} \right)}{\partial
        {g}_{l,i}} = \mathcal{R}\left\{ \left(\mathbf {\bar P} \mathbf{h} - \mathbf{y}\right)^{H} \frac{\partial
        \mathbf {\bar P} \mathbf{h} }{\partial
        { g}_{l,i}} \right\},
\end{eqnarray}
where
\begin{align}
 \frac{\partial \mathbf {\bar P} \mathbf{h} }{\partial{g}_{l,i}}
&= \frac{\partial \mathbf {\bar P} \sum_{l=1}^{L}\mathbf{q}_{\rm NU}\left( \mathbf g_{l}, \mathbf f_{l} \right) \sigma_{l}  }{\partial{g}_{l,i}} = \sigma_{l}\mathbf{\bar P}\frac{\partial \mathbf q_{\rm NU}\left( \mathbf g_{l}, \mathbf f_{l} \right)}{\partial g_{l,i}}, \\
 \frac{\partial \mathbf{q}_{\rm NU}(\mathbf{g}_l,\mathbf{f}_l)}{\partial
        g_{l,i}} &= \frac{\partial \mathbf{a}_{\mathbf d_{t}}^{*}\left(\mathbf{
                g}_{l}\right) \otimes \mathbf{ b}_{\mathbf d_{r}}\left(\mathbf{
                f}_{l}\right)}{\partial g_{l,i}} = \frac{\partial \mathbf{a}_{\mathbf
                d_{t}}^{*}\left(\mathbf{g}_{l}\right)}{\partial g_{l,i}}  \otimes \mathbf{b}_{\mathbf d_{r}}\left(\mathbf{f}_{l}\right),\\ \label{eq:differential_g_l}
\frac{\partial \mathbf{a}_{\mathbf d_{t}}^{*}\left(\mathbf{g}_{l}\right)}{\partial
g_{l,i}} &=  \left( \frac{-j{2\pi}}{\lambda} \left[ d_{t,i} (1), \ldots, d_{t,i} (N) \right]^{T}\right) \circ \mathbf{a}_{\mathbf d_{t}}^{*}\left(\mathbf{g}_{l}\right). \end{align}
By plugging \eqref{eq:differential_g_l} into \eqref{eq:derive_subgradient_g}, we have \eqref{eq:subgradient_g}. Similarly we can obtain \eqref{eq:subgradient_f}.

\subsection{ Derivation for \eqref{eq:subgradient_sigma} }
The gradient with respect to $\sigma_{l}$ can be calculated by
\begin{eqnarray} \label{eq:diff_sigma}
\nabla_{ \sigma_{l}} \Gamma\left( \{ \mathbf g_{l}, \mathbf f_{l}, \sigma_{l} \} \right) &=&\frac{\partial \left( \mu\left\Vert \bm{\sigma} \right\Vert_{1}
        + \frac{1}{2} \left\Vert \mathbf{y} - \mathbf {\bar P} \mathbf{h}  \right\Vert^{2}_{2} \right) }{\partial {\sigma}_{l}^*}  \nonumber\\
&=& \frac{\partial \|\boldsymbol{\sigma}\|_1}{\partial {\sigma}_{l}^*} - \frac{1}{2} \frac{\partial \mathbf{y}^H \mathbf {\bar P} \mathbf h}{\partial
        {\sigma}_{l}^*} -
\frac{1}{2} \frac{\partial \mathbf{h}^H  \mathbf {\bar P}^H\mathbf{y}}{\partial
        {\sigma}_{l}^*} + 
        \frac{1}{2}\frac{\partial \mathbf{h}^H \mathbf {\bar P}^H
        \mathbf {\bar P}\mathbf{h} }{\partial
        {\sigma}_{l}^*},
\end{eqnarray}
where
\begin{align}
\frac{\partial \mathbf{y}^H \mathbf {\bar P}\mathbf{h} }{\partial
        {\sigma_{l}}^*} &=
\mathbf{y}^H \mathbf {\bar P} \mathbf q_{\rm NU} \left( \mathbf g_{l}, \mathbf f_{l} \right)
\frac{\partial {\sigma}_{l}}{\partial
       {\sigma}_{l}^*} ={0}, \label{eq:diff_sigma_step_1}\\
\frac{\partial \mathbf{h}^H  \mathbf {\bar P}^H\mathbf{y}}{\partial
        {\sigma}_{l}^*} &=
\left(\left( \mathbf {\bar P}
\mathbf{q}_{\rm NU} \left( \mathbf{g}_{l}, \mathbf{f}_{l}  \right)\right)^H\mathbf{y}\right)^T\frac{\partial
        {\sigma}_{l}^*}{\partial{\sigma}_{l}^*}=
\left(\left( \mathbf {\bar P}
\mathbf{q}_{\rm NU} \left( \mathbf{g}_{l}, \mathbf{f}_{l}  \right)\right)^H\mathbf{y}\right)^T, \label{eq:diff_sigma_step_2} \\
\frac{\partial \mathbf{h}^H \mathbf {\bar P}^H
        \mathbf {\bar P}\mathbf{h}}{\partial
        {\sigma}_{l}^*} & =\left( \left(\mathbf {\bar P}
\mathbf{q}_{\rm NU} \left( \mathbf{g}, \mathbf{f}  \right)  \right)^H
\mathbf {\bar P}
\mathbf{q}_{\rm NU} \left( \mathbf{ g}_{l}, \mathbf{f}_{l}  \right){\sigma}_{l}\right)^T, \label{eq:diff_sigma_step_3}\\
\frac{\partial \|\boldsymbol{\sigma}\|_1}{\partial \sigma_l^*} &= \frac{\partial
        \sum_l|\sigma_l|}{\partial \sigma_l^*}= \frac{\partial |\sigma_l|}{\partial
        \sigma_l^*}=\frac{1}{2}\left(\frac{\partial|\sigma_l| }{\partial \mathcal{R}\{\sigma_l\}}+i\frac{\partial|\sigma_l|
}{\partial \mathcal{I}\{\sigma_l\}}\right) \label{eq:diff_sigma_step_4}\\
& =\frac{1}{2}\left(\frac{\partial
        \sqrt{
                \mathcal{R}^2\{\sigma_l\}+ \mathcal{I}^2\{\sigma_l\}
} }{\partial \mathcal{R}\{\sigma_l\}}+i\frac{\partial \sqrt{
                \mathcal{R}^2\{\sigma_l\}+ \mathcal{I}^2\{\sigma_l\}
} }{\partial \mathcal{I}\{\sigma_l\}}\right)\notag\\
& =\frac{\sigma_l}{2|\sigma_l| }, \label{eq:diff_sigma_step_5} \notag
\end{align}
where $\mathcal{I}\left\{ \cdot \right\}$ returns the imaginary part of the input.
Plugging  \eqref{eq:diff_sigma_step_1}-\eqref{eq:diff_sigma_step_4} into \eqref{eq:diff_sigma}, we  obtain \eqref{eq:subgradient_sigma}.

\bibliography{bfchest_refs}

\begin{thebibliography}{10}
\providecommand{\url}[1]{#1}
\csname url@samestyle\endcsname
\providecommand{\newblock}{\relax}
\providecommand{\bibinfo}[2]{#2}
\providecommand{\BIBentrySTDinterwordspacing}{\spaceskip=0pt\relax}
\providecommand{\BIBentryALTinterwordstretchfactor}{4}
\providecommand{\BIBentryALTinterwordspacing}{\spaceskip=\fontdimen2\font plus
\BIBentryALTinterwordstretchfactor\fontdimen3\font minus
  \fontdimen4\font\relax}
\providecommand{\BIBforeignlanguage}[2]{{%
\expandafter\ifx\csname l@#1\endcsname\relax
\typeout{** WARNING: IEEEtran.bst: No hyphenation pattern has been}%
\typeout{** loaded for the language `#1'. Using the pattern for}%
\typeout{** the default language instead.}%
\else
\language=\csname l@#1\endcsname
\fi
#2}}
\providecommand{\BIBdecl}{\relax}
\BIBdecl

\bibitem{5G_it_work}
T.~S. Rappaport, S.~Sun, R.~Mayzus, H.~Zhao, Y.~Azar, K.~Wang, G.~N. Wong,
  J.~K. Schulz, M.~Samimi, and F.~Gutierrez, ``Millimeter wave mobile
  communications for 5g cellular: It will work!'' \emph{IEEE Access}, vol.~1,
  pp. 335--349, 2013.

\bibitem{3D-MIMO}
X.~Cheng, B.~Yu, L.~Yang, J.~Zhang, G.~Liu, Y.~Wu, and L.~Wan, ``Communicating
  in the real world: 3d mimo,'' \emph{IEEE Wireless Communications}, vol.~21,
  no.~4, pp. 136--144, August 2014.

\bibitem{mmWave_BF_for_backhaul_and_small_cell}
S.~Hur, T.~Kim, D.~J. Love, J.~V. Krogmeier, T.~A. Thomas, and A.~Ghosh,
  ``Millimeter wave beamforming for wireless backhaul and access in small cell
  networks,'' \emph{{IEEE} Trans. Commun.}, vol.~61, no.~10, pp. 4391--4403,
  October 2013.

\bibitem{FD-MIMO_next_gen_cell_tech}
Y.~H. Nam, B.~L. Ng, K.~Sayana, Y.~Li, J.~Zhang, Y.~Kim, and J.~Lee,
  ``Full-dimension mimo (fd-mimo) for next generation cellular technology,''
  \emph{IEEE Communications Magazine}, vol.~51, no.~6, pp. 172--179, June 2013.

\bibitem{Beamforming_Survey}
S.~Kutty and D.~Sen, ``Beamforming for millimeter wave communications: An
  inclusive survey,'' \emph{IEEE Communications Surveys Tutorials}, vol.~18,
  no.~2, pp. 949--973, Secondquarter 2016.

\bibitem{Beam_codebook_vased_beamforming_protocol_for_mmWave}
J.~Wang, ``Beam codebook based beamforming protocol for multi-gbps
  millimeter-wave wpan systems,'' \emph{IEEE Journal on Selected Areas in
  Communications}, vol.~27, no.~8, pp. 1390--1399, October 2009.

\bibitem{3D_mmWave_channel_model}
M.~K. Samimi and T.~S. Rappaport, ``3-d millimeter-wave statistical channel
  model for 5g wireless system design,'' \emph{IEEE Transactions on Microwave
  Theory and Techniques}, vol.~64, no.~7, pp. 2207--2225, July 2016.

\bibitem{3D_mmWave_channel_model_proposal}
T.~A. Thomas, H.~C. Nguyen, G.~R. MacCartney, and T.~S. Rappaport, ``3d mmwave
  channel model proposal,'' in \emph{2014 IEEE 80th Vehicular Technology
  Conference (VTC2014-Fall)}, Sept 2014, pp. 1--6.

\bibitem{Spatially_Sparse-Precoding_in_MM_Wave}
O.~E. Ayach, S.~Rajagopal, S.~Abu-Surra, Z.~Pi, and R.~W. Heath, ``Spatially
  sparse precoding in millimeter wave mimo systems,'' \emph{{IEEE} Trans.
  Wireless Commun.}, vol.~13, no.~3, pp. 1499--1513, March 2014.

\bibitem{Channel_Estimation_and_Hybrid_Precoding}
A.~Alkhateeb, O.~E. Ayach, G.~Leus, and R.~W. Heath, ``Channel estimation and
  hybrid precoding for millimeter wave cellular systems,'' \emph{IEEE Journal
  of Selected Topics in Signal Processing}, vol.~8, no.~5, pp. 831--846, Oct
  2014.

\bibitem{CS_based_multi_users_mmWave}
A.~Alkhateeby, G.~Leusz, and R.~W. Heath, ``Compressed sensing based multi-user
  millimeter wave systems: How many measurements are needed?'' in \emph{2015
  IEEE International Conference on Acoustics, Speech and Signal Processing
  (ICASSP)}, April 2015, pp. 2909--2913.

\bibitem{2D_MUSIC_mmWave}
Z.~Guo, X.~Wang, and W.~Heng, ``Millimeter-wave channel estimation based on
  two-dimensional beamspace music method,'' \emph{IEEE Transactions on Wireless
  Communications}, vol.~PP, no.~99, 2017.

\bibitem{Millimeter_Wave_MIMO_channel_estimation_Adaptive_CS}
S.~Sun and T.~S. Rappaport, ``Millimeter wave mimo channel estimation based on
  adaptive compressed sensing,'' in \emph{2017 IEEE International Conference on
  Communications Workshops (ICC Workshops)}, May 2017, pp. 47--53.

\bibitem{Donoho_ComprssedSensing}
D.~L. Donoho, ``Compressed sensing,'' \emph{{IEEE} Trans. Inform. Theory},
  vol.~52, no.~4, pp. 1289--1306, April 2006.

\bibitem{book_mmWav_Communications}
T.~S. Rappaport, R.~W. Heath, R.~C. Daniels, and J.~N. Murdock,
  \emph{Millimeter Wave Wireless Communication}.\hskip 1em plus 0.5em minus
  0.4em\relax Prentice-Hall, 2014.

\bibitem{Compressed_Sensing_Off_the_Grid}
G.~Tang, B.~N. Bhaskar, P.~Shah, and B.~Recht, ``Compressed sensing off the
  grid,'' \emph{{IEEE} Trans. Inform. Theory}, vol.~59, no.~11, pp. 7465--7490,
  Nov 2013.

\bibitem{Bhaskar_AtomicNormDenoise}
B.~N. Bhaskar, G.~Tang, and B.~Recht, ``Atomic norm denoising with applications
  to line spectral estimation,'' \emph{{IEEE} Trans. Sig. Proc.}, vol.~61,
  no.~23, pp. 5987--5999, 2013.

\bibitem{Gridless_super_resolution_direction_finding}
J.~Steinwandt, F.~Roemer, C.~Steffens, M.~Haardt, and M.~Pesavento, ``Gridless
  super-resolution direction finding for strictly non-circular sources based on
  atomic norm minimization,'' in \emph{2016 50th Asilomar Conference on
  Signals, Systems and Computers}, Nov 2016, pp. 1518--1522.

\bibitem{Atomic_Norm_Denoising}
B.~N. Bhaskar, G.~Tang, and B.~Recht, ``Atomic norm denoising with applications
  to line spectral estimation,'' \emph{{IEEE} Trans. Sig. Proc.}, vol.~61,
  no.~23, pp. 5987--5999, Dec 2013.

\bibitem{ANM_DoA_estimation}
Y.~Zhang, G.~Zhang, and X.~Wang, ``Array covariance matrix-based atomic norm
  minimization for off-grid coherent direction-of-arrival estimation,'' in
  \emph{2017 IEEE International Conference on Acoustics, Speech and Signal
  Processing (ICASSP)}, March 2017, pp. 3196--3200.

\bibitem{ANM_2D_DoA_estimation}
Z.~Tian, Z.~Zhang, and Y.~Wang, ``Low-complexity optimization for
  two-dimensional direction-of-arrival estimation via decoupled atomic norm
  minimization,'' in \emph{2017 IEEE International Conference on Acoustics,
  Speech and Signal Processing (ICASSP)}, March 2017, pp. 3071--3075.

\bibitem{ANM_MIMO_CHEST}
P.~Zhang, L.~Gan, S.~Sun, and C.~Ling, ``Atomic norm denoising-based channel
  estimation for massive multiuser mimo systems,'' in \emph{2015 IEEE
  International Conference on Communications (ICC)}, June 2015, pp. 4564--4569.

\bibitem{ANM_Linear_System_Identification}
P.~Shah, B.~N. Bhaskar, G.~Tang, and B.~Recht, ``Linear system identification
  via atomic norm regularization,'' in \emph{2012 IEEE 51st IEEE Conference on
  Decision and Control (CDC)}, Dec 2012, pp. 6265--6270.

\bibitem{structured_non_unifromly_spaced_antenna_array}
W.~Liu, Z.~Wang, C.~Sun, S.~Chen, and L.~Hanzo, ``Structured non-uniformly
  spaced rectangular antenna array design for fd-mimo systems,'' \emph{{IEEE}
  Trans. Wireless Commun.}, vol.~16, no.~5, pp. 3252--3266, May 2017.

\bibitem{NUPA_mmwave}
P.~Wang, Y.~Li, Y.~Peng, S.~C. Liew, and B.~Vucetic, ``Non-uniform linear
  antenna array design and optimization for millimeter-wave communications,''
  \emph{IEEE Transactions on Wireless Communications}, vol.~15, no.~11, pp.
  7343--7356, Nov 2016.

\bibitem{FD-MIMO_Using_Uniform_Planar_Arrays}
J.~Choi, K.~Lee, D.~J. Love, T.~Kim, and R.~W. Heath, ``Advanced limited
  feedback designs for fd-mimo using uniform planar arrays,'' in \emph{2015
  IEEE Global Communications Conference (GLOBECOM)}, Dec 2015, pp. 1--6.

\bibitem{3D_MIMO_ULA}
Y.~Han, S.~Jin, X.~Li, Y.~Huang, L.~Jiang, and G.~Wang, ``Design of double
  codebook based on 3d dual-polarized channel for multiuser mimo system,''
  \emph{EURASIP Journal on Advances in Signal Processing}, vol. 2014, no.~1, p.
  111, 2014.

\bibitem{DFT_beamforming}
D.~Yang, L.~L. Yang, and L.~Hanzo, ``Dft-based beamforming weight-vector
  codebook design for spatially correlated channels in the unitary precoding
  aided multiuser downlink,'' in \emph{2010 IEEE International Conference on
  Communications}, May 2010, pp. 1--5.

\bibitem{Exploiting_spatial_sparsity_for_estimating_channels_of_hybrid_MIMO_systems_in_mmWave}
J.~Lee, G.~T. Gil, and Y.~H. Lee, ``Exploiting spatial sparsity for estimating
  channels of hybrid mimo systems in millimeter wave communications,'' in
  \emph{2014 IEEE Global Communications Conference}, Dec 2014, pp. 3326--3331.

\bibitem{MUSIC_DoA_Estimation}
P.~Gupta and S.~P. Kar, ``Music and improved music algorithm to estimate
  direction of arrival,'' in \emph{2015 International Conference on
  Communications and Signal Processing (ICCSP)}, April 2015, pp. 0757--0761.

\bibitem{Joint_Transmission_Reception_Diversity_Smoothing}
W.~Zhang, W.~Liu, J.~Wang, and S.~Wu, ``Joint transmission and reception
  diversity smoothing for direction finding of coherent targets in mimo
  radar,'' \emph{IEEE Journal of Selected Topics in Signal Processing}, vol.~8,
  no.~1, pp. 115--124, Feb 2014.

\bibitem{Exact_Joint_Sparse}
Z.~Yang and L.~Xie, ``Exact joint sparse frequency recovery via optimization
  methods,'' \emph{{IEEE} Trans. Sig. Proc.}, vol.~64, no.~19, pp. 5145--5157,
  Oct 2016.

\bibitem{Vandermonde_Decomposition_Multilevel_Toeplitz_Matrices}
Z.~Yang, L.~Xie, and P.~Stoica, ``Vandermonde decomposition of multilevel
  toeplitz matrices with application to multidimensional super-resolution,''
  \emph{{IEEE} Trans. Inform. Theory}, vol.~62, no.~6, pp. 3685--3701, June
  2016.

\bibitem{Compressive_Recovery_2D_off_grid}
Y.~Chi and Y.~Chen, ``Compressive recovery of 2-d off-grid frequencies,'' in
  \emph{2013 Asilomar Conference on Signals, Systems and Computers}, Nov 2013,
  pp. 687--691.

\bibitem{ChiC15}
------, ``Compressive two-dimensional harmonic retrieval via atomic norm
  minimization,'' \emph{{IEEE} Trans. Sig. Proc.}, vol.~63, no.~4, pp.
  1030--1042, 2015.

\bibitem{candes2014towards}
E.~J. Cand{\`e}s and C.~Fernandez-Granda, ``Towards a mathematical theory of
  super-resolution,'' \emph{Communications on Pure and Applied Mathematics},
  vol.~67, no.~6, pp. 906--956, 2014.

\bibitem{Super-Resolution-Radar}
L.~Zheng and X.~Wang, ``Super-resolution delay-doppler estimation for ofdm
  passive radar,'' \emph{{IEEE} Trans. Sig. Proc.}, vol.~65, no.~9, pp.
  2197--2210, May 2017.

\bibitem{2D_AN_approx}
Z.~Tian, Z.~Zhang, and Y.~Wang, ``Low-complexity optimization for
  two-dimensional direction-of-arrival estimation via decoupled atomic norm
  minimization,'' in \emph{2017 IEEE International Conference on Acoustics,
  Speech and Signal Processing (ICASSP)}, March 2017, pp. 3071--3075.

\bibitem{boyd2011distributed}
S.~Boyd, N.~Parikh, E.~Chu, B.~Peleato, and J.~Eckstein, ``Distributed
  optimization and statistical learning via the alternating direction method of
  multipliers,'' \emph{Foundations and Trends{\textregistered} in Machine
  Learning}, vol.~3, no.~1, pp. 1--122, 2011.

\bibitem{sturm1999using}
J.~F. Sturm, ``Using sedumi 1.02, a matlab toolbox for optimization over
  symmetric cones,'' \emph{Optimization methods and software}, vol.~11, no.
  1-4, pp. 625--653, 1999.

\bibitem{toh1999sdpt3}
K.-C. Toh, M.~J. Todd, and R.~H. Tutuncu, ``Sdpt3-a matlab software package for
  semidefinite programming, version 1.3,'' \emph{Optimization methods and
  software}, vol.~11, no. 1-4, pp. 545--581, 1999.

\bibitem{NUPA_mmWaveMIMOLinks}
E.~Torkildson, C.~Sheldon, U.~Madhow, and M.~Rodwell, ``Nonuniform array design
  for robust millimeter-wave mimo links,'' in \emph{GLOBECOM 2009 - 2009 IEEE
  Global Telecommunications Conference}, Nov 2009, pp. 1--7.

\bibitem{nai2010beampattern}
S.~E. Nai, W.~Ser, Z.~L. Yu, and H.~Chen, ``Beampattern synthesis for linear
  and planar arrays with antenna selection by convex optimization,'' \emph{IEEE
  Transactions on Antennas and Propagation}, vol.~58, no.~12, pp. 3923--3930,
  2010.

\bibitem{li2016approximate}
Q.~Li and G.~Tang, ``Approximate support recovery of atomic line spectral
  estimation: A tale of resolution and precision,'' in \emph{Signal and
  Information Processing (GlobalSIP), 2016 IEEE Global Conference on}.\hskip
  1em plus 0.5em minus 0.4em\relax IEEE, 2016, pp. 153--156.

\bibitem{boumal2015low}
N.~Boumal and P.-A. Absil, ``Low-rank matrix completion via preconditioned
  optimization on the grassmann manifold,'' \emph{Linear Algebra and its
  Applications}, vol. 475, pp. 200--239, 2015.

\bibitem{OMP}
T.~T. Cai and L.~Wang, ``Orthogonal matching pursuit for sparse signal recovery
  with noise,'' \emph{{IEEE} Trans. Inform. Theory}, vol.~57, no.~7, pp.
  4680--4688, July 2011.

\bibitem{cvx}
M.~Grant and S.~Boyd, ``{CVX}: Matlab software for disciplined convex
  programming, version 2.1,'' \url{http://cvxr.com/cvx}, Mar. 2014.

\end{thebibliography}

%\newpage
%Figures

\printindex

\end{document}